%% file: paper.tex
\documentclass[letterpaper,11pt]{article}

\usepackage{makeidx}  % allows for indexgeneration

\usepackage{amsthm}
\usepackage[margin=1in,dvips]{geometry}

\usepackage{graphicx}
\usepackage{amssymb}
\usepackage{amsmath}
\usepackage{float}
\usepackage{enumerate}
\usepackage{url}
\usepackage{subfig}
\usepackage[hidelinks]{hyperref}

\newtheorem{theorem}{Theorem}[section]
\newtheorem{lemma}[theorem]{Lemma}
\newtheorem{corollary}[theorem]{Corollary}

\newtheorem{observation}[theorem]{Observation}

\newcommand{\tail}[2]{#1_{\overline{[#2]}}}
\newcommand{\abs}[1]{\left|#1\right|}

\newcommand{\norm}[2]{\lVert#2\rVert_{#1}}

\newcommand{\wh}{\widehat}

\newcommand{\PE}{E_p}
\newcommand{\KE}{E_k}

\input{theorems.tex}

{\makeatletter
 \gdef\xxxmark{%
   \expandafter\ifx\csname @mpargs\endcsname\relax % in minipage?
     \expandafter\ifx\csname @captype\endcsname\relax % in figure/caption?
       \marginpar{xxx}% not in a caption or minipage, can use marginpar
     \else
       xxx % notice trailing space
     \fi
   \else
     xxx % notice trailing space
   \fi}
 \gdef\xxx{\@ifnextchar[\xxx@lab\xxx@nolab}
 \long\gdef\xxx@lab[#1]#2{{\bf [\xxxmark #2 ---{\sc #1}]}}
 \long\gdef\xxx@nolab#1{{\bf [\xxxmark #1]}}
 \long\gdef\xxx@lab[#1]#2{}\long\gdef\xxx@nolab#1{}%
}

\DeclareMathOperator*{\mean}{mean}
\DeclareMathOperator{\supp}{supp}
\DeclareMathOperator*{\median}{median}
\DeclareMathOperator*{\E}{\mathbb{E}}
\DeclareMathOperator*{\Var}{Var}
\DeclareMathOperator*{\Cov}{Cov}
\def\R{\mathbb{R}}

\def\F{\mathcal{F}}
\def\Thetat{\widetilde{\Theta}}
\def\eps{\epsilon}

\makeatletter 
\newenvironment{proof*}[1][\proofname]{\par 
  \normalfont \partopsep=\z@skip \topsep=\z@skip 
  \trivlist 
  \item[\hskip\labelsep 
        \itshape 
    #1\@addpunct{.}]\ignorespaces 
}{% 
  \endtrivlist\@endpefalse 
} 
\makeatother 

\begin{document}

\begin{titlepage}
\title{Improved Concentration Bounds for
    Count-Sketch\thanks{This work began when both authors were funded by
      internships at Microsoft Research.  GM received further support
      from Hertz Foundation and National Science Foundation Fellowships,
      and EP received further support from a Simons Fellowship.}}

\date{}
\author{Gregory T. Minton\\MIT \and Eric Price\\MIT}
\maketitle

\begin{abstract}
  We present a refined analysis of the classic Count-Sketch streaming
  heavy hitters algorithm~\cite{CCF02}.  Count-Sketch uses $O(k \log
  n)$ linear measurements of a vector $x \in \R^n$ to give an estimate
  $\wh{x}$ of $x$.  The standard analysis shows that this estimate
  $\wh{x}$ satisfies $\norm{\infty}{\wh{x}-x}^2 <
  \norm{2}{\tail{x}{k}}^2 / k$, where $\tail{x}{k}$ is the vector
  containing all but the largest $k$ coordinates of $x$.  Our main
  result is that most of the coordinates of $\wh{x}$ have
  substantially less error than this upper bound; namely, for any $c <
  O(\log n)$, we show that each coordinate $i$ satisfies
\[
(\wh{x}_i - x_i)^2 < \frac{c}{\log n} \cdot \frac{||\tail{x}{k}||_2^2}{k}
\]
with probability $1-2^{-\Omega(c)}$, as long as the hash functions are
fully independent.  This subsumes the previous bound and is optimal
for all $c$.  Using these improved point estimates, we prove a stronger
concentration result for set estimates by first analyzing the covariance
matrix and then using a median-of-median-of-medians argument to
bootstrap the failure probability bounds.  These results also give
improved results for $\ell_2$ recovery of exactly $k$-sparse estimates
$x^*$ when $x$ is drawn from a distribution with suitable decay, such
as a power law or lognormal.

We complement our results with simulations of Count-Sketch on a power
law distribution.  The empirical evidence indicates that our theoretical
bounds give a precise characterization of the algorithm's performance:
the asymptotics are correct and the associated constants are small.

Our proof shows that any symmetric random variable with finite variance
and positive Fourier transform concentrates around $0$ at least as well as a
Gaussian.  This result, which may be of independent interest, gives good
concentration even when the noise does not converge to a Gaussian.
\end{abstract}

\thispagestyle{empty}
\end{titlepage}

\section{Introduction}

The \emph{heavy hitters} problem and the closely related \emph{sparse
  recovery} problem are two of the most fundamental problems in the
field of sketching and streaming
algorithms~\cite{CCF02,CM06,GI10,CH10,M05}.  The goal is to
efficiently identify and estimate the $k$ largest coordinates of an
$n$-dimensional vector using a linear sketch $Ax$ of $x$, where $A \in
\R^{m\times n}$ has $m = O(k \log^c n)$ rows.  The strongest commonly
used formal guarantee for the quality of such an estimate is the
$\ell_\infty/\ell_2$ guarantee: this is a bound for the estimate $\wh
x$ recovered from $Ax$ which is of the form
\begin{align}\label{e:linfl2}
  \norm{\infty}{\wh{x}-x}^2 \leq \norm{2}{\tail{x}{k}}^2 / k,
\end{align}
where $\tail{x}{k}$ denotes the vector obtained from $x$ by replacing
its largest $k$ coordinates with $0$.

The classic approach for this problem is the Count-Sketch algorithm of
Charikar et al.~\cite{CCF02}, which uses $m = O(k \log n)$
measurements and satisfies~\eqref{e:linfl2} with $1 - 1/n^{\Theta(1)}$
probability.  It is simple, practical, and gives the best known
theoretical performance in many settings.  It also pioneered a
technique---hashing with random signs and estimating using
medians---that forms the basis for several subsequent works on sparse
recovery~\cite{GLPS,IPW11,HIKP12a,G12}.

\paragraph{Our result.} We show that, despite the popularity of
Count-Sketch, its performance has not been fully characterized and
understood. Specifically, we prove that the quality of the
approximation $\wh{x}$ given by Count-Sketch is better than the
standard bound~\eqref{e:linfl2} suggests.  While~\eqref{e:linfl2}
gives a bound on the \emph{worst-case} error of $\wh{x}$, we prove
that \emph{most} coordinates of $\wh{x}$ have asymptotically smaller
error than this worst case.

The Count-Sketch of a vector $x$ using $R$ rows of $C$
columns is defined as follows.  For $u \in [R]$, we choose hash functions $h_u:
[n] \to [C]$ and $s_u: [n] \to \{\pm 1\}$.  The sketch is
\[
y_{u,v} = \sum_{i : h_u(i) = v} s_u(i)x_i,
\]
which consists of $RC$ linear measurements.  The estimate $\wh{x}$ is given by
\[
\wh{x}_i = \median_u s_u(i)y_{u,h_u(i)}.
\]
Setting $C = O(k)$ and $R = O(\log n)$,~\cite{CCF02}
proves that~\eqref{e:linfl2} holds with $1-1/n^{\Theta(1)}$ probability.
\iffalse
As per~\cite{CM06,GI10}, the $\ell_\infty/\ell_2$ guarantee can be converted
into an $\ell_2/\ell_2$ guarantee: if $x^*$ contains the largest $k$
coordinates of $\wh{x}$, then
\[
\norm{2}{x^*-x}^2 \leq O(\norm{2}{\tail{x}{k}}^2).
\]
\fi

%The Count-Sketch matrix has also been shown to satisfy the
%Johnson-Lindenstrauss property with $(R, C) = (\frac{1}{\eps}\log
%(1/\delta), \frac{1}{\eps})$~\cite{KN12}.

%\paragraph{Our Results}
Our main result is the following strengthening of the analysis
in~\cite{CCF02} for the accuracy of the \emph{point estimates}
$\wh{x}_i$ resulting from Count-Sketch, assuming the hash functions
are fully random:

\define{thm:main}{Theorem}{%
  Consider the estimate $\wh{x}$ of $x$ from Count-Sketch using $R$
  rows and $k \geq 2$ columns, with fully random hash functions.  For
  any $t \leq R$ and each index $i$,
  \[
  \Pr\left[(\wh{x}_i - x_i)^2 > \frac{t}{R} \cdot
    \frac{\norm{2}{\tail{x}{k}}^2}{k}\right] < 2e^{-\Omega(t)}.
  \]
}
\state{thm:main}

The standard analysis~\cite{CCF02} proves this bound in the special
case of $t=R$; one then gets~\eqref{e:linfl2} by setting $t = R =
\Theta(\log n)$ and applying a union bound.  We show in
Theorem~\ref{thm:lower} that our stronger result is optimal; it gives the
\emph{best possible} failure probability for all $t \leq O(\log(n/k))$
and all linear sketches.

Theorem~\ref{thm:main} shows that the \emph{average} squared error
of a set is $1/R$ times the previously known bound, i.e., the bound coming
from \eqref{e:linfl2}.  We extend this in Theorem~\ref{thm:bettersets} to show
concentration when estimating a set of coordinates, so that the total squared error over
the set satisfies our improved bound \emph{with high probability}.

\paragraph{Implications.}
Often, one performs Count-Sketch in order to estimate the largest $k$
coordinates of $x$.  In this case, the bound~\eqref{e:linfl2} gives an
optimal result for arbitrary vectors $x$~\cite{PW11} but not
necessarily for common distributions on $x$.  A particularly important
distribution is the \emph{power law} or \emph{Zipfian} distribution,
which is the standard distribution to analyze for sparse
recovery~\cite{CCF02,CM05,CRT06,BCDH}.  Consider again Count-Sketch with $R
= \Theta(\log n)$ rows.  We show that if $x$ follows the power law
$x_i = i^{-\alpha}$ for some constant $\alpha > 0.5$, then the error
in estimating the largest $k$ coordinates is $1/\log n$ times the
previously known bound with high probability (see Theorem~\ref{thm:l2}
for details).  The same result holds for other common distributions
such as lognormals or exponentials.

Previous work~\cite{P11} combined Count-Sketch with another sketch to
get the same bound as in Theorem~\ref{thm:l2}, but our result here applies
directly to the output of Count-Sketch.  This is important because
Count-Sketch is an algorithm that is used in practice, while chains of
algorithms are less likely to be used---especially because years later
we may discover that the original algorithm performed as well as the
chain!  For example, Google uses Count-Sketch to estimate the largest
$k$ coordinates of $x$ for their ``\texttt{top} table'', a core
language feature of their MapReduce programming language
Sawzall~\cite{PDGQ}.  Because many datasets Google encounters (for
example, the frequency of URLs on the web) are distributed as power
laws or lognormals~\cite{M04,BKM+,CM05}, Theorem~\ref{thm:l2} directly
applies to their setting.

\paragraph{Experiments.}
Finally, we complement our analysis with simulations of Count-Sketch
on a power law distribution.  These show that, unlike previous results,
Theorem~\ref{thm:main} and Theorem~\ref{thm:l2} correctly characterize
the asymptotic performance of point and top-$k$ estimates,
respectively.  Furthermore, the constants involved are small:
between $1/2$ and $2$.  We also find that Count-Sketch has
asymptotically less error than Count-Min, an alternative sketch
algorithm.

\paragraph{Limitations.}
Our analysis requires that the hash functions be fully random.
This is unfortunate because fully random hash functions would
take up more space than the sketch itself, but there are reasons
why this constraint is not too problematic.  One reason is that
Nisan's pseudorandom number generator~\cite{N92} lets us store the
hash
functions with only a $\log n$ factor increase in space.  Then
if we wish to run Count-Sketch on multiple different vectors, we can reuse
the hash functions.
A second reason is that one expects bounded independence to suffice
as long as the vector $x$ itself has sufficient entropy.  A result of
this form is known~\cite{MV10} when $\supp(x)$ is drawn at random from
a much larger domain.  For example, if $\supp(x)$ contains $n^{1/3}$
random coordinates out of $n$, then~\cite{MV10} implies
near-uniformity with 4-wise independence.

\paragraph{Our Techniques}  
Our basic strategy is to translate the problem of bounding
Count-Sketch error into a problem of proving a strong concentration
result for a certain class of random variables.  This, in
turn, we solve by analyzing the Fourier transform of such variables.

In more detail, the argument proceeds as follows.  The error
$\wh x_i - x_i$ is, by definition, the median over rows of
error terms coming from the different coordinates which hash
to the same column as $i$.  For each row, we separate the error term
into contributions from (i) the largest $k$ coordinates $j \in [k]$
and (ii) the remaining coordinates $j \in \overline{[k]}$.
The error of type (i) is zero with constant probability, and
we bound the error of type (ii) with our concentration result.
We then get a bound on $\wh x_i - x_i$ by using Chernoff bounds
to conclude that if each of $R$ symmetric random variables has a $\sqrt{c/R}$
chance of being small, then the median has a $1-2e^{-c/2}$ chance of being small.
This proves strong bounds for the error of point estimates;
we then analyze the pairwise dependence of said errors to
conclude a bound on the error of sets.

The concentration result we prove is a bound of the form $\Pr[|X| <
\epsilon] > \Omega(\epsilon)$, where $X$ has variance 1 and is a sum
of independent random variables, each of which is symmetric and zero
with probability at least $1/2$ (Corollary~\ref{cor:gen}).
Such a bound certainly holds in the limit as
$X$ converges to a Gaussian, but we need it to be true before $X$
converges. To see why this is subtle, consider the sum of $n$
independent $\pm 1/\sqrt{n}$ variables.  The Berry-Ess\'een theorem
gives our bound for $\eps > 1/\sqrt{n}$, but the bound is actually
false for $\eps < 1/\sqrt{n}$ when $n$ is odd.  When $n$ is even,
we can pair up the variables to get $n/2$ independent
$\{0, \pm2/\sqrt{n}\}$ variables.  These variables are zero with $1/2$
probability, so our bound applies for arbitrarily small $\eps$.
What distinguishes even $n$ from odd $n$?

The key for our argument is that, for a symmetric random variable $X$
with at least $1/2$ probability of being $0$, the Fourier transform
of $X$ is \emph{nonnegative}.
The Fourier transform of the triangle filter $\max \{1-|x|/\epsilon,0\}$
is also nonnegative.  We use the convolution theorem to translate
the expectation of the triangle filter into an integral in Fourier
space, and then use positivity to note that we can bound that
integral over all Fourier space by the integral over small
frequencies.  This we control directly by using the quadratic Taylor series
approximation to $\cos x$.  Because a lower bound on
the expectation of the triangle filter also gives a lower bound
on $\Pr[|X| < \epsilon]$, this proves what we want.

The above techniques let us prove Theorem~\ref{thm:main}, which shows
that, for Count-Sketch with $O(\log n)$ rows and $k$ columns, the
squared error in \emph{point estimates} of individual coordinates $i$
is exponentially distributed with mean $\mu^2/\log n$, where $\mu^2$
is the previously known bound.

We generally want to estimate multiple coordinates at a time, though,
so we proceed to bound the average error over \emph{sets} of
coordinate estimates.  It follows easily from Theorem~\ref{thm:main}
that the average error is $\mu^2/\log n$ in expectation; however, one
might expect to get this error with high probability, since averages
tend to concentrate as the size of the set grows.  Getting strong
concentration is difficult because the errors in different coordinates
are not independent.  To handle this, we resort to the following
approach.  Consider sets of size $k$.  We first show that the error
coming from collisions with small coordinates can be replaced by
independent noise, and then we define a variant of Count-Sketch which
is pairwise independent.  By bounding the difference of regular
Count-Sketch and this pairwise independent variant, we get a bound on the
covariance matrix of the errors for each coordinate in our set.
We then apply Chebyshev's inequality, getting $\mu^2/\log n$
error with failure probability $O(k^{-1/14})$
(Proposition~\ref{prp:sets}).  This bound is
nontrivial but falls well short of the ``high probability'' standard
of $O(1/k^c)$ failure probability for arbitrary constant $c$.
Unfortunately, while a more refined bound on the covariance matrix
could improve the exponent, no approach based on Chebyshev's inequality
can prove better than a $1/\sqrt{k}$ failure probability.

However, there is a kludge that gives the $\Thetat(1/k^c)$ failure
probability we want.  Consider running $O(c)$ Count-Sketches in parallel
and taking the (coordinate-wise) median of the results of each Count-Sketch.
Some analysis shows that this boosts the failure probability from
$\Thetat(k^{-1/14})$ to the desired $\Thetat(1/k^c)$ (Corollary~\ref{c:hack}).
Our goal, though, is to analyze the simple Count-Sketch algorithm that
people actually use instead of this hackish variant.  Notice that the kludge uses the
same set of measurements as Count-Sketch with an $O(c)$ factor more rows,
but then performs recovery by estimating each coordinate as a median
(over chunks) of medians (within chunks), rather than Count-Sketch's direct
medians.  To complete the argument we show, via our ``Median$^{3}$ Lemma''
(Lemma~\ref{l:median3}), that taking medians directly cannot be much worse
than computing the median of medians.  Thus true Count-Sketch also satisfies
the desired $\Thetat(1/k^c)$ bound on the failure probability
(Theorem~\ref{thm:bettersets}); in summary, the weak bound we get from
bounding the covariance matrix bootstraps into a better bound.

% *sigh*
%Implications for electoral colleges vs direct elections are left as an
%open question.

\section{Preliminaries}

\paragraph{Notation}
We use $f \gtrsim g$ to denote $f = \Omega(g)$ and $f \lesssim g$ to denote $f = O(g)$.

In the statement of Theorem \ref{thm:main},
$\tail{x}{k}$ denotes the vector consisting of all but
the largest $k$ coordinates of $x$.  More generally,
we think of the coordinates of $x$ as being sorted,
$\abs{x_1} \ge \abs{x_2} \ge \cdots \ge \abs{x_n}$.
This is purely a notational convenience, possible because Count-Sketch is
invariant under permutation of coordinates.

Given a real-valued random variable $X$, its
\emph{Fourier transform} is the function
\[\F(t) = \E[e^{2 \pi X t \sqrt{-1}}].\]
In general $\F(t)$ is complex-valued.  However, our random
variables are all symmetric; in this case $\F(t)$ is
real-valued and equals $\E[\cos(2\pi X t)]$.

\section{Concentration Lemmas}

The following is the key lemma for our proof.

\begin{lemma}\label{lemma:gen}
  Let $X$ be a symmetric, real-valued random variable
  with variance $1$, and suppose that its Fourier transform $\F(t)$ is nonnegative.
  Then, for $\eps \le 1$, $\Pr[\abs{X} < \eps] \gtrsim \eps$.
\end{lemma}
\begin{proof}
Because $\cos x \ge 1 - \tfrac{1}{2} x^2$ holds for all $x \in \R$,
we have
\[\F(t) \ge \E[1 - \tfrac{1}{2} (2 \pi X t)^2] = 1 - 2 \pi^2 t^2\ \ \ \forall\ t\in\R.\]
In particular, $\F(t) \ge \tfrac{1}{2}$ for $t \in [-\tfrac{1}{2\pi},\tfrac{1}{2\pi}]$.
Let $T_\epsilon(x)$ be the triangle filter
\[T_\epsilon(x) = \begin{cases} 1 - \tfrac{1}{\epsilon}|x| & \text{if }|x| < \epsilon \\ 0 & \text{otherwise} \end{cases}\]
and recall the Fourier transform relation
\[T_\epsilon(x) = \int_{-\infty}^\infty \frac{\sin^2(\pi t \epsilon)}{\pi^2 t^2 \epsilon} e^{2 \pi x t \sqrt{-1}}\ \text{d}t.\]
Using this relation and switching the order of integration,
\[\E[T_\epsilon(X)] = \int_{-\infty}^\infty \frac{\sin^2(\pi t \epsilon)}{\pi^2 t^2 \epsilon} \F(t)\,\text{d}t.\]
The integrand is nonnegative, so we get a lower bound on $\E[T_\epsilon(X)]$
by integrating only over the interval $[-\tfrac{1}{2\pi},\tfrac{1}{2\pi}]$.  On this
interval we have $\F(t) \ge \tfrac{1}{2}$ and,
because $\epsilon < \pi$,
$\frac{\sin^2(\pi t \epsilon)}{\pi^2 t^2 \epsilon}$
is bounded below by its
value at $t = 1/(2\pi)$.  Putting this together, we find that
\[\E[T_\epsilon(X)] \ge \int_{-1/(2\pi)}^{1/(2\pi)} \frac{\sin^2(\pi t \epsilon)}{\pi^2 t^2 \epsilon} \F(t)\,\text{d}t \gtrsim \frac{\sin^2(\epsilon/2)}{\epsilon}.\]
For $\epsilon \le 1$ we have $\tfrac{\sin^2(\epsilon/2)}{\epsilon} \gtrsim \epsilon$.
Now noting that $\Pr[|X| < \epsilon] \ge \E[T_\epsilon(X)]$ completes
the proof.
\end{proof}

\begin{corollary} \label{cor:gen}
Let $\{X_i : i \in [n]\}$ be independent symmetric random variables such that
$\Pr[X_i = 0] \ge 1/2$ for each $i$.
Set $X = \sum_{i=1}^n X_i$ and $\sigma^2 = \E[X^2]$.
For $\epsilon \le 1$, $\Pr[|X| < \epsilon \sigma] \gtrsim \epsilon$.
\end{corollary}
\begin{proof}
For each $i \in [n]$, let $p_i = \Pr[X_i = 0]$.  The
Fourier transform of $X_i$ is
$\F_i(t) = p_i + (1-p_i) \E[\cos(2 \pi X_i t) \mid X_i \ne 0] \ge p_i + (1-p_i)(-1)$.
Because $p_i \ge 1/2$, this is nonnegative.  Now $X/\sigma$ is a symmetric random
variable with nonnegative Fourier transform $\prod_{i=1}^n \F_i(t/\sigma)$ and
with variance $\E[(X/\sigma)^2] = 1$; applying Lemma \ref{lemma:gen}
to it gives the desired bound.
\end{proof}

%Careful examination of the proof of Lemma~\ref{lemma:gen}
%shows that the constant can be taken to be $2 \sin^2(1/2)/\pi \approx 0.146$.
%In particular, we can take
%$\Pr[|X| < \epsilon \sigma] \ge \tfrac{1}{7}\epsilon$
%in Corollary~\ref{cor:gen}.

Note that Lemma \ref{lemma:gen} is not true without the
positivity assumption; in particular, as we observed in the
introduction, Corollary \ref{cor:gen} is not true when
$\Pr[X_i = 0]$ is small.  Indeed, it seems
intuitive that we get strong concentration
around $0$ as a consequence of the large probability of each individual
variable being $0$.  We also remark that there are analogs of Lemma
\ref{lemma:gen} and Corollary \ref{cor:gen} using only first moment
bounds.  The proof is nearly identical, so we omit it.

We also need the following lemma for concentration of medians.

\begin{lemma}\label{lemma:medians}
  Suppose $X_1, \dotsc, X_t$ are independent symmetric random variables such
  that, for some $x,p > 0$, we have
  $\Pr[\abs{X_i} < x] > p$ for all $i \in [t]$.  Then
  \[
  \Pr\left[\abs{\median_{i \in [t]} X_i} \geq x\right] < 2e^{-t p^2/2}.
  \]
\end{lemma}
\begin{proof}
  Let $E_i$ denote the indicator for the event that $X_i \geq x$.
  Because $X_i$ is symmetric we have
  $\Pr[E_i = 1] < (1 - p) / 2$, so $\E[\sum_{i=1}^t E_i] < t/2 - pt/2$.
  The $E_i$ are independent, so by a Chernoff bound we have that
  \[
  \Pr\left[\sum_{i=1}^t E_i \geq \frac{t}{2}\right] < e^{-2(pt/2)^2/t} = e^{-t p^2/2}.
  \]
  The same bound applies to the event that at least $t/2$ of the $X_i$
  are less than $-x$, and if neither event occurs then the median is in
  the interval $(-x,x)$.
\end{proof}

\section{Count-Sketch}

\declare{thm:main}
\begin{proof}
  Fix $i \in [n]$.  For each row $u$ and coordinate $j \in [n]$, define
  \[
  X_{u,j} = \left\{
    \begin{array}{cl}
      s_u(j)x_j & \text{if }h_u(j) = h_u(i)\\
      0 & \text{otherwise.}
    \end{array}
  \right.
  \]
  For each row $u$, define
  \[T_u = \sum_{j \in \overline{[k]}\setminus \{i\}} X_{u,j}\ \ \ \text{and}\ \ \ H_u = \sum_{j \in [k]\setminus \{i\}} X_{u,j}.\]
  Then, by definition,
  \[
  \wh{x}_i - x_i = \median_u H_u + T_u.
  \]
  Each random variable $X_{u,j}$ is symmetric, equals $0$ with
  probability $1 - 1/k \ge 1/2$, and otherwise equals $\pm x_j$.
  Moreover, for each row $u$, the random variables $\{X_{u,j}\}$ are independent.
  Thus $\E[T_u^2] = \sum_{j \in \overline{[k]}\setminus\{i\}} x_j^2/k \le \norm{2}{\tail{x}{k}}^2/k$, so
  Corollary~\ref{cor:gen} shows that
  \[
  \Pr\left[\abs{T_u} < \epsilon \cdot \frac{\norm{2}{\tail{x}{k}}}{\sqrt{k}}\right] \gtrsim \epsilon
  \]
  for all $\epsilon \le 1$.
  Furthermore, $H_u = 0$ with probability at least $(1-1/k)^k \ge 1/4$,
  i.e., with constant probability.
  Since $H_u$ is independent of $T_u$, this means that
  \[
  \Pr\left[\abs{H_u + T_u} < \epsilon \cdot \frac{\norm{2}{\tail{x}{k}}}{\sqrt{k}}\right] \gtrsim \epsilon.
  \]
  Therefore Lemma~\ref{lemma:medians} implies
  \[
  \Pr\left[\abs{\wh{x}_i - x_i} > \epsilon \cdot \frac{\norm{2}{\tail{x}{k}}}{\sqrt{k}}\right] < 2e^{-\Omega(R \epsilon^2)}.
  \]
  Setting $\epsilon = \sqrt{t/R}$ yields the desired result.
\end{proof}

\section{Concentration for Sets} \label{sec:sets}

Theorem~\ref{thm:main} shows that each individual error
$(\wh{x}_i-x)^2$ has a constant chance of being less than $O(1/R)$
times the $\ell_\infty^2$ bound.  One would reasonably suspect that
the average error over large sets would satisfy this bound with high
probability.  This is in fact true.  The following result is proven in
Appendix~\ref{app:sets}.

% EP wants to set t=1 and drop it, on the grounds that changing
% k can absorb a constant.  GM finds that displeasing, since what's
% really happening is that we absorb the constant by making C bigger
% and, while it all works out fine, it's not clear how to express
% that in a technically precise way.  I (GM) wants to keep the t
% here... I drop it later in the final theorem.

\define{prp:sets}{Proposition}{%
  Fix a constant $t > 0$ and consider the estimate $\wh{x}$ of $x$ from Count-Sketch
  using $R$ rows and $C = ck$ columns, $\log k \lesssim R$, for sufficiently large
  (depending on $t$) constant $c$.  For any set $S \subset [n]$ with $\abs{S} \leq k$,
  \begin{align*}
  \Pr\left[\norm{2}{\wh{x}_S - x_S}^2 > t \cdot \abs{S} \cdot \frac{1}{R} 
    \cdot \frac{\norm{2}{\tail{x}{k}}^2}{k} \right] \lesssim \frac{1}{\abs{S}^{1/14}}.
  \end{align*}
}
\declare{prp:sets}

The analysis leading to Proposition~\ref{prp:sets} is excessively lossy but, as we will see
presently, we can improve the resulting bound after the fact so that
the loss is only temporary.

\section{Improving the Probability Bound}
\label{sec:mediancubed}

To get a better bound on the failure probability than
Proposition~\ref{prp:sets}, we first consider the procedure of
running Count-Sketch a constant number of times in parallel
and taking the median of the resulting estimates.  Using
this procedure lets us improve the exponent in the failure
probability to any desired constant.

\begin{lemma} \label{l:medians}
Let $x^{(1)},\dotsc,x^{(r)}$ be vectors in $\R^n$ and let
$x \in \R^n$ be the coordinate-wise median of
$\{x^{(1)},\dotsc,x^{(r)}\}$.  If at least a $3/4$ fraction
of the variables $x^{(i)}$ satisfy $\norm{2}{x^{(i)}} < C$,
then $\norm{2}{x} < C\sqrt{3}$.
\end{lemma}
\begin{proof}
Choose $3r/4$ indices $i$ satisfying $\norm{2}{x^{(i)}} < C$;
call these indices ``good''.  Fix a coordinate $j$.
For at least $r/2$ indices $i$ we have $x_j \le x^{(i)}_j$
and for at least $r/2$ we have $x_j \ge x^{(i)}_j$;
thus (using the first group if $x_j \ge 0$ and the
second group if $x_j < 0$) for at least $r/2$ indices
$i$ we have $x_j^2 \le (x^{(i)}_j)^2$.  Of these, at least
$r/2-r/4 = r/4$ must also be good.  Hence
\[x_j^2 \le \mean_{\substack{\text{good $i$ s.t.}\\x_j^2 \le (x^{(i)}_j)^2}} (x^{(i)}_j)^2 \le \frac{1}{r/4} \sum_{\substack{\text{good $i$ s.t.}\\x_j^2 \le (x^{(i)}_j)^2}} (x^{(i)}_j)^2 \le \frac{1}{r/4} \sum_{\text{good $i$}} (x^{(i)}_j)^2 = 3 \mean_{\text{good $i$}} (x^{(i)}_j)^2.\]
Summing over the coordinates $j$ gives
$\norm{2}{x}^2 \le 3 \mean_{\text{(good $i$)}} \norm{2}{x^{(i)}}^2 < 3C^2$.
\end{proof}

We remark in passing that there is a generalization of
Lemma~\ref{l:medians} in which one replaces Euclidean
balls with convex, coordinate-wise symmetric sets.

\begin{lemma} \label{l:hack}
Suppose $\{X_1,\dots,X_r\}$ are independent random variables
taking values in $\R^n$.  Let $X$ be the random variable
obtained by taking the coordinate-wise median of
$\{X_1,\dots,X_r\}$.  If $\Pr[\norm{2}{X_i} < C] \ge 1-p$
for each $i$, then $\Pr[\norm{2}{X} < C\sqrt{3}] \ge 1-(11p)^{r/4}$.
\end{lemma}
\begin{proof}
  Let $E_i$ denote the event that $\norm{2}{X_i} \ge C$.  The
  probability that at least $r/4$ of the $E_i$ occur is at most
  $\binom{r}{r/4}p^{r/4} \leq (4ep)^{r/4} \leq (11p)^{r/4}$.  Thus,
  with probability at least $1-(11p)^{r/4}$, at least a $3/4$ fraction
  of the variables $X_i$ satisfy $\norm{2}{X_i} < C$.  When this holds
  we have $\norm{2}{X} < C\sqrt{3}$ by Lemma~\ref{l:medians}.
\end{proof}

\begin{corollary} \label{c:hack} Fix a real constant $t>0$
  and a positive integer $d$ and consider the
  estimate $\wh{x}$ of $x$ coming from running $56d$ instances of
  Count-Sketch in parallel, each using $R$ rows and $C = ck$ columns
  (for sufficiently large --- depending $d$ and $t$ --- constant $c$), and
  then taking the coordinate-wise median of the $56d$ resulting
  estimates.  Suppose $\log k \lesssim R$.  For any set $S \subset [n]$ with $\abs{S} \leq k$,
  \begin{align*}
  \Pr\left[\norm{2}{\wh{x}_S - x_S}^2 > t \cdot \abs{S} \cdot \frac{1}{R} 
    \cdot \frac{\norm{2}{\tail{x}{k}}^2}{k} \right] \lesssim \frac{1}{\abs{S}^d}.
  \end{align*}
\end{corollary}
\begin{proof}
  For $i \in [56d]$ let $\wh{x}_i$ denote the estimate from
  the $i$th instance of Count-Sketch.  Using Proposition~\ref{prp:sets}
  we can choose $c$ such that
  \begin{align*}
  \Pr\left[\norm{2}{(\wh{x}_{i})_S - x_S}^2 >
    \frac{t}{3} \cdot \abs{S} \cdot \frac{1}{R} 
    \cdot \frac{\norm{2}{\tail{x}{k}}^2}{k} \right] \lesssim
  \frac{1}{\abs{S}^{1/14}}.
  \end{align*}
   We now get the desired result by applying Lemma \ref{l:hack} to
   the random variables $\{(\wh{x}_i - x)_S\}_i$.
\end{proof}

We now conclude the section by showing that the bound in
Corollary~\ref{c:hack} applies to Count-Sketch itself.
The key is the following combinatorial observation, which
can be summarized as ``the median of the median-of-medians is the median!''

\begin{lemma}[Median${}^{\text{3}}$] \label{l:median3} Let
  $\{a_1,\dots,a_n\}$ be a list of $n=k \ell$ real numbers with $n$
  odd.  Consider the set $\Pi$ of all partitions $\pi =
  \{S_1,\dots,S_\ell\}$ of $[n]$ into blocks of size $k$.  Then
  \[
  \median_{\pi\in\Pi} \median_{b \in [\ell]} \median_{i \in S_b} a_i = \median_{i \in [n]} a_i.
  \]

\end{lemma}
\begin{proof}
  As medians depends only on the relative orderings, without loss of
  generality we may assume that the set $\{a_i\}$ is symmetric about
  $0$ (e.g., take $a_i = -(n+1)/2 + i$).  Both sides of the
  desired equality are invariant under permutation of coordinates;
  hence they are both invariant under negation of the elements $a_i$
  and so are both zero.
\end{proof}

\begin{theorem} \label{thm:bettersets} Fix a constant $d$, and
  consider the estimate $\wh{x}$ of $x$ from Count-Sketch using $R$
  rows and $C = ck$ columns, for sufficiently large (depending on $d$)
  constant $c$.  Suppose $\log k \lesssim R$.  For any set $S \subset [n]$
  with $\abs{S} \leq k$,
  \begin{align*}
  \Pr\left[\norm{2}{\wh{x}_S - x_S}^2 > \abs{S} \cdot \frac{1}{R} 
    \cdot \frac{\norm{2}{\tail{x}{k}}^2}{k} \right] \lesssim
  \frac{1}{\abs{S}^d}.
  \end{align*}
\end{theorem}
\begin{proof}
  Let $\pi$ be a partition of $[R]$ into $56d$ blocks of
  size $R/(56d)$ and let $\wh{x}_\pi$ denote the estimate obtained by
  running Count-Sketch separately on each block and then taking the
  median of the results (as in Corollary~\ref{c:hack}).  Define
\[B =  \abs{S} \cdot \frac{1}{R} \cdot \frac{\norm{2}{\tail{x}{k}}^2}{k}\]
  and let $E_\pi$ be the indicator for the event
  $\norm{2}{(\wh{x}_\pi - x)_S}^2 > \tfrac{1}{3}B$.
  Define
\[p = 4 \Pr[E_\pi = 1] = 4 \Pr \left [ \norm{2}{(\wh{x}_\pi - x)_S}^2 > \frac{1}{3} \cdot \abs{S} \cdot \frac{1}{R} \cdot \frac{\norm{2}{\tail{x}{k}}^2}{k} \right ].\]
 By Corollary~\ref{c:hack}, we can choose the constant $c$ so that
  $p \lesssim 1/\abs{S}^d$.

  This holds for any partition $\pi$.  Letting $N$ denote
  the number of such partitions, we have $\E[\sum_\pi E_\pi] \le Np/4$,
  and so $\Pr[\sum_\pi E_\pi > \tfrac{1}{4}N] < p$ by Markov's
  inequality.  Suppose now (as happens with at least $1-p$ probability)
  that $\sum_\pi E_\pi \le N/4$.  Then, letting $\widetilde x$ be the 
  coordinate-wise median of $\wh{x}_\pi$ over all partitions $\pi$,
  we have $\norm{2}{(\widetilde x - x)_S}^2 \le B$ by
  Lemma~\ref{l:medians}.  But
  $\widetilde x = \wh{x}$ by the Median$^{3}$ Lemma (Lemma~\ref{l:median3}).
  Putting this together, we have $\norm{2}{\wh{x}_S - x_S}^2 \le B$
  with probability at least $1-p$, which is exactly what we wanted.
\end{proof}

\section{Concentration for Compressible Signals}\label{sec:powerlaw}

One key application of Count-Sketch is to compute a table estimating
the largest $k$ coordinates of $x$~\cite{PDGQ}.  Some questions arise
about the proper metric for evaluating such estimates.  For continuous
distributions, distinguishing the $k$th and $(k+1)$st largest
coordinates is both difficult and not very important.  We choose to
measure the ``distance to validity,'' meaning the distance from $x$ to
the nearest $x'$ which has the same top $k$ coordinates as $\wh{x}$.
That is, if $H_k(x)$ denotes the restriction of $x$ to its $k$ largest
components, then we denote the ``top-$k$ estimation error'' of
$\wh{x}$ by
\begin{align}
  \label{e:KE}
  \KE(x, \wh{x}) := \min_{x' : H_k(x') = H_k(\wh{x})} \norm{2}{x - x'}.
\end{align}
The basic $\ell_{\infty}/\ell_2$ guarantee~\eqref{e:linfl2} gives that,
with $R = \Theta(\log n)$ and $C = O(k)$, Count-Sketch satisfies $\KE
\lesssim \norm{2}{x_{\overline{[k]}}}^2$.  By~\cite{PW11}, this is
optimal on worst-case inputs $x$.

However, real-world signals are not worst-case.  In fact, signals are
likely to be well approximated by power law or lognormal
distributions~\cite{M04,BKM+,CM05}, and sparsity is mainly useful
because such signals are, in fact, sparse~\cite{CRT06,BCDH}.

In this section we consider recovery of signals with suitable decay:
that is, signals where $\abs{x_k} - \abs{x_{2k}} \gtrsim
\norm{2}{\tail{x}{k}}/\sqrt{k}$.  This condition is satisfied by any
power law distribution $x_i \approx i^{-\alpha}$ with $\alpha > 0.5$, which
is the range of $\alpha$ for which the distribution is sparse (in
$\ell_2$); the condition is also satisfied by lognormal distributions
in the range for which they are sparse.

We show that, for such signals, $\KE \lesssim
\norm{2}{x_{\overline{[k]}}}^2 / R$ with high probability.  This gives
a factor of $R$ improvement over the standard result.  The idea is that
while Theorem~\ref{thm:bettersets} only applies to fixed sets of
indices, on such distributions the largest $k$ coordinates of $\wh{x}$
will, with high probability, be among the largest $2k$ coordinates of
$x$.  Hence we can apply Theorem~\ref{thm:bettersets} to that fixed
set of $2k$ coordinates.

\define{thm:l2}{Theorem}{%
  Suppose $\abs{x_k} - \abs{x_{2k}}
  \gtrsim \norm{2}{\tail{x}{k}}/\sqrt{k}$ and fix a constant $d$.  Let
  $\wh{x}$ be the result of Count-Sketch using $R \gtrsim \log n$ rows
  and $\Theta(k)$ columns, with fully random hash functions and
  constant factors depending on $d$.  Define $\KE$ as in~\eqref{e:KE}.  Then
  \[
  \KE(x, \wh{x}) \leq \frac{1}{R}\norm{2}{x_{\overline{[k]}}}^2
  \]
  with $1 - O(1/k^d)$ probability.
}
\declare{thm:l2}
\begin{proof}
  Let the number of columns be $ck$ for some constant $c$.  By the
  standard Count-Sketch bound we have with $1-n^{-\Theta(1)}$
  probability that $\norm{\infty}{\wh{x}-x}^2 <
  \norm{2}{\tail{x}{ck}}^2/(ck)$.  Then for sufficiently large $c$,
  \begin{align}
   \abs{\wh{x}_i} > \max(\abs{\wh{x}_j}, \abs{x_j})\label{e:12}
  \end{align}
  for all $i \in [k]$ and $j \in \overline{[2k]}$.

  Let $x'$ equal $\wh{x}$ over $[2k]$ and $x$ over $\overline{[2k]}$.
  Then by~\eqref{e:12} the top $k$ coordinates of $\wh{x}$ and of $x'$
  both lie among $[2k]$; since $x' = \wh{x}$ on this region, the top
  $k$ coordinates of the two are equal.  Hence
  \[
  \KE(x, \wh{x}) \leq \norm{2}{x - x'} = \norm{2}{\wh{x}_{[2k]} - x_{[2k]}}^2.
  \]
  But by Theorem~\ref{thm:bettersets},
  \[
  \norm{2}{\wh{x}_{[2k]} - x_{[2k]}}^2 \lesssim
  \frac{1}{R}\frac{2k}{ck}\norm{2}{\tail{x}{2k}}^2
  \]
  with probability at least $1 - O(1 / k^d)$.  Setting $c$ large
  enough gives the result.
\end{proof}

\section{Lower Bound on Point Queries}\label{sec:lower}

The following is an application of the proof technique
of~\cite{PW11}, using Gaussian channel capacity to bound the number of
measurements required for a given error tolerance.

\define{thm:lower}{Theorem}{
  For any $1 \leq t \leq \log(n/k)$ and any distribution on $O(Rk)$ linear
  measurements of $x \in \R^n$, there is some vector $x$ and
  index $i$ for which the estimate $\wh{x}$ of $x$ satisfies
  \[
  \Pr\left[(\wh{x}_i-x_i)^2 > \frac{t}{R}
  \frac{\norm{2}{\tail{x}{k}}^2}{k}\right] > e^{-\Omega(t)}.
  \]
}
\declare{thm:lower}
\begin{proof}
  Suppose without loss of generality that $n = k2^t$ (by ignoring
  indices outside $[k2^t]$) and that $t$ is larger than some constant.
  Partition $[n]$ into $k$ blocks of size $2^t$.  Set $x = y + w$,
  where $y \in \{0, 1, -1\}^n$ has a single random $\pm 1$ in each
  block (so it is $k$-sparse) and $w = N(0, \eps\frac{Rk}{nt}I_n)$ for
  some constant $\eps$ is i.i.d.\ Gaussian.

  Suppose that, in expectation over $x$, $A \in \R^{m \times n}$
  allows recovering $\wh{x}$ from $Ax$ with
  \begin{align}\label{e:goodestimate}
    (\wh{x}_i-x_i)^2 \leq \frac{t}{R} \frac{\norm{2}{\tail{x}{k}}^2}{k}
  \end{align}
  for more than a $1 - 2^{-2t}$ fraction of the coordinates $i$.
  We will show that such an $A$ must have $m
  \gtrsim Rk$ rows.  Yao's minimax principle then gives a lower bound
  for distributions on $A$.  With this, the inability to increase $t$ and $k$
  while preserving the number of rows gives the desired lower
  bound on failure probability.

  First, we show that $I(Ax; z) \gtrsim k t$.  Let $E$ be the event
  that \eqref{e:goodestimate} holds for more than a $1 - 2^{-t-2}$
  fraction of coordinates $i$ and that $\norm{2}{w}^2 < 2
  \E[\norm{2}{w}^2] = 2\eps Rk / t$.  $E$ holds with probability
  $1 - o(1) > 1/2$ probability over $x$.  Conditioned
  on $E$, we have
  \[
  (\wh{x}_i-x_i)^2 < 2\eps
  \]
  for a $1 - 2^{-t-2}$ fraction of the coordinates $i$.  Thus, for
  $\eps = 1/8$, if we round $\wh{x}_i$ to the nearest integer we
  recover $x^*$ with $x^*_i = z_i$ in a $1 - 2^{-t-2}$ fraction of the
  coordinates; hence $x^*_i = z_i$ over at least $3/4$ of the blocks.
  We know that $z$ has $(t+1)$ bits of entropy in each block.  This
  means, conditioned on $E$,
  \begin{align*}
    I(z; x^*) &= H(z) - H(z \mid x^*) \\&\geq k(t+1) - \log(\binom{k}{k/4}
    2^{(t+1)k/4}) \\&\geq k(t+1) - k(t+1)/4 - k\log(4e)/4 \gtrsim kt
  \end{align*}
  and hence $I(Ax; z \mid E=1) \gtrsim kt$
  by the data processing inequality.  But since $\Pr[E] \geq 1/2$,
  \begin{align}\label{e:Ilower}
    I(Ax; z) &\geq I(Ax; z \mid E) - H(E) \geq I(Ax; z \mid E=1) \Pr[E]
    - 1 
    \notag\\&
    \gtrsim kt.
  \end{align}
  Second, we show that $I(Ax; z) \lesssim mt/R$.  For each row
  $A_j$, $A_jx = A_jz + A_jw = A_jz + w'$ for $w' \sim N(0,
  \norm{2}{A_j}^2\eps Rk/(nt))$.  We also have $\E_z[(A_jz)^2] =
  \norm{2}{A}^2k/n$.  Hence $A_jx$ is an additive white Gaussian noise
  channel with signal-to-noise ratio
  \[
  \frac{\E[(A_jz)^2]}{\E[w'^2]} = \frac{t}{\eps R}.
  \]
  By the Shannon-Hartley Theorem, this channel has capacity
  \[
  I(A_jx; z) \leq \frac{1}{2}\log(1 + \frac{t}{\eps R}) <
  \frac{t}{\eps R} \lesssim t/R
  \]
  and thus, by linearity and independence of $w'$ (as in~\cite{PW11}),
  \begin{align}\label{e:Iupper}
    I(Ax; z) \lesssim mt/R
  \end{align}
  Combining~\eqref{e:Ilower} and \eqref{e:Iupper} gives $m \gtrsim
  Rk$.
\end{proof}

\section{Simulation}\label{sec:simul}

Theorems~\ref{thm:main} and~\ref{thm:l2} give asymptotic upper bounds
on the error of Count-Sketch estimates.  Theorem~\ref{thm:lower}
shows that there exists a distribution on inputs for which
Theorem~\ref{thm:main} gives the correct asymptotics.  However, this
does not show that the asymptotics are correct on common input
distributions, or that these asymptotics appear at practical input
sizes.

To address these questions, in this section we discuss empirical results demonstrating
that, on the most common model of input distributions,
\begin{itemize}
\item Theorems~\ref{thm:main} and~\ref{thm:l2} give the right asymptotics;
\item the constants involved are small; and
\item the estimates are better than those of Count-Min, an alternative
  estimation algorithm.
\end{itemize}

\subsection{Simulation Details}
We draw $x \in \R^n$ from the Pareto (Type I) distribution with
parameter $\alpha = 1.25$, chosen because Pareto distributions are common in large data sets and $\alpha \in [1, 1.5]$ is typical~\cite{CSN,M04}.
This distribution is given by
\[
\Pr[x_i > t] = (\mu/t)^\alpha
\]
independently for each $i$, where the scaling parameter
\[
\mu = n^{-1/\alpha}\sqrt{2/\alpha - 1}
\]
is chosen so that $\E[\norm{2}{x_{\overline{[k]}}}^2] \approx
k^{1-2/\alpha} = k^{-0.6}$.  (Note that, for $k \geq 10$ and large $n$, the
error in the approximation $\E[\norm{2}{x_{\overline{[k]}}}^2] \approx
k^{1-2/\alpha}$ is less than $1\%$.)

We then perform Count-Sketch with $R$ rows and $C$ columns, for
various $R$ and $C$, to get estimates $\wh{x}$ of $x$.  We will
analyze the distributions of point error and top-$k$ estimation error,
as distributions over $x$ and the Count-Sketch.  Point error is
defined as
\[
\PE = \abs{\wh{x}_i - x_i}
\]
for a random coordinate $i$.  For top-$k$ estimation error $\KE$, we use the
definition~\eqref{e:KE} from \textsection\ref{sec:powerlaw}.

We will study the behavior of $\KE$ and $\PE$ for large $n$ as a
function of $R$, $C$, and $k$, in order to empirically verify the
following specific claims.

\begin{itemize}
\item (Theorem~\ref{thm:main}) After removing $n2^{-\Omega(R)}$
  probability mass, the point estimation error $\abs{\wh{x}_i - x_i}$
  has expectation
  \[
  \E[\PE] \simeq \frac{\E[\norm{2}{x_{\overline{[C]}}}]}{\sqrt{RC}}  \simeq \frac{\E[\norm{2}{x_{\overline{[C]}}}^2]^{1/2}}{\sqrt{RC}} \approx
  \frac{1}{R^{.5}C^{.8}} =: m_{R,C}
  \]
  and decays like a Gaussian:
  \[
  \Pr[\abs{\wh{x}_i - x_i} > tm_{R,C}] \leq e^{-\Omega(t^2)}.
  \]
\captionsetup{margin=1em}
\begin{figure}[H]
  \centering \subfloat[Distribution of $\PE/m_{R,C}$ for various $R,
  C$.  Note that it is nearly independent of $R$ and
  $C$.]{\label{f:ped}
    \includegraphics[width=0.5\textwidth]{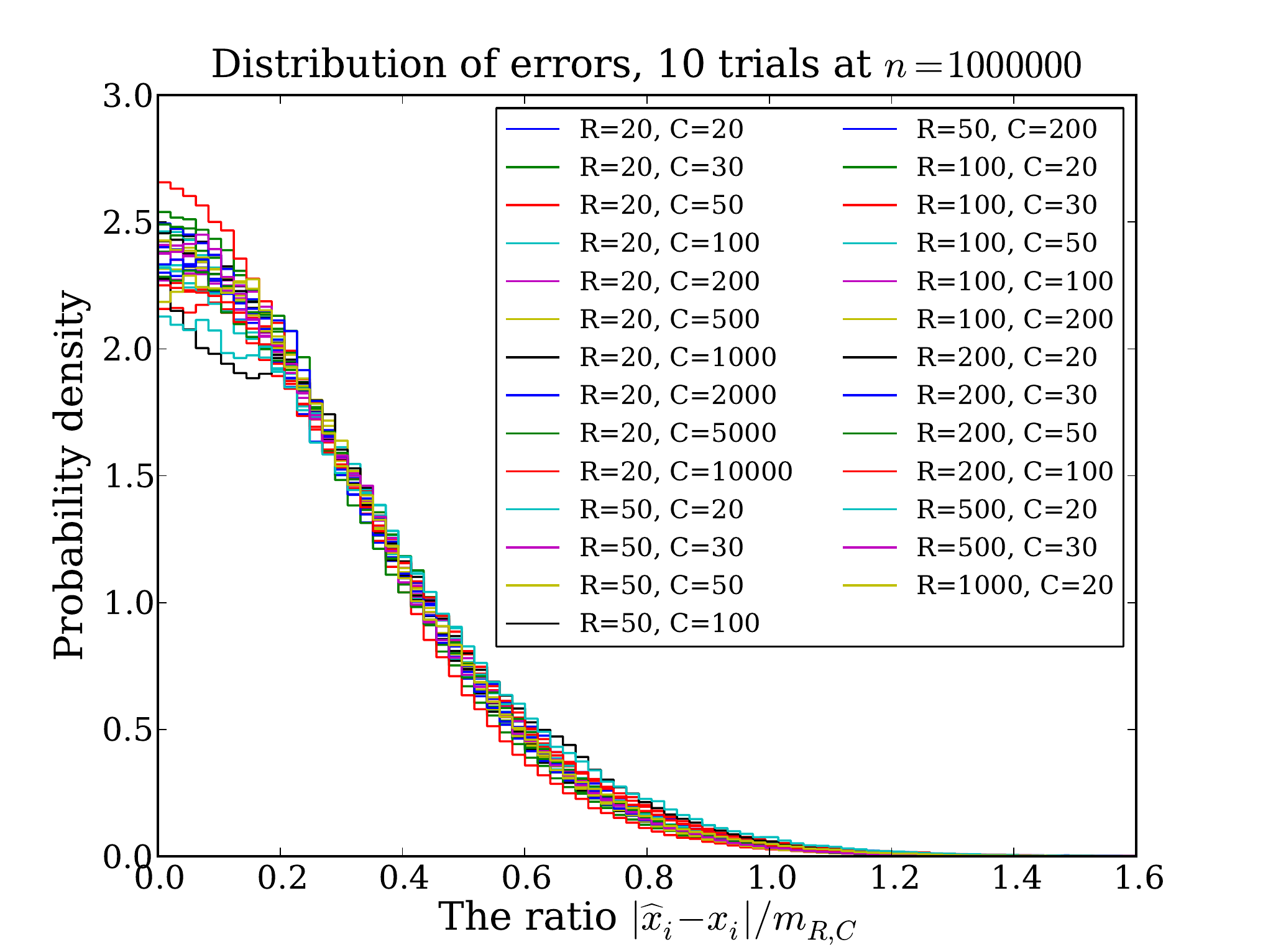}
  } \subfloat[Same as~(a), but with Count-Min added for comparison.
  Note that Count-Min has larger error than Count-Sketch.]{
    \label{f:pecm}
    \includegraphics[width=0.5\textwidth]{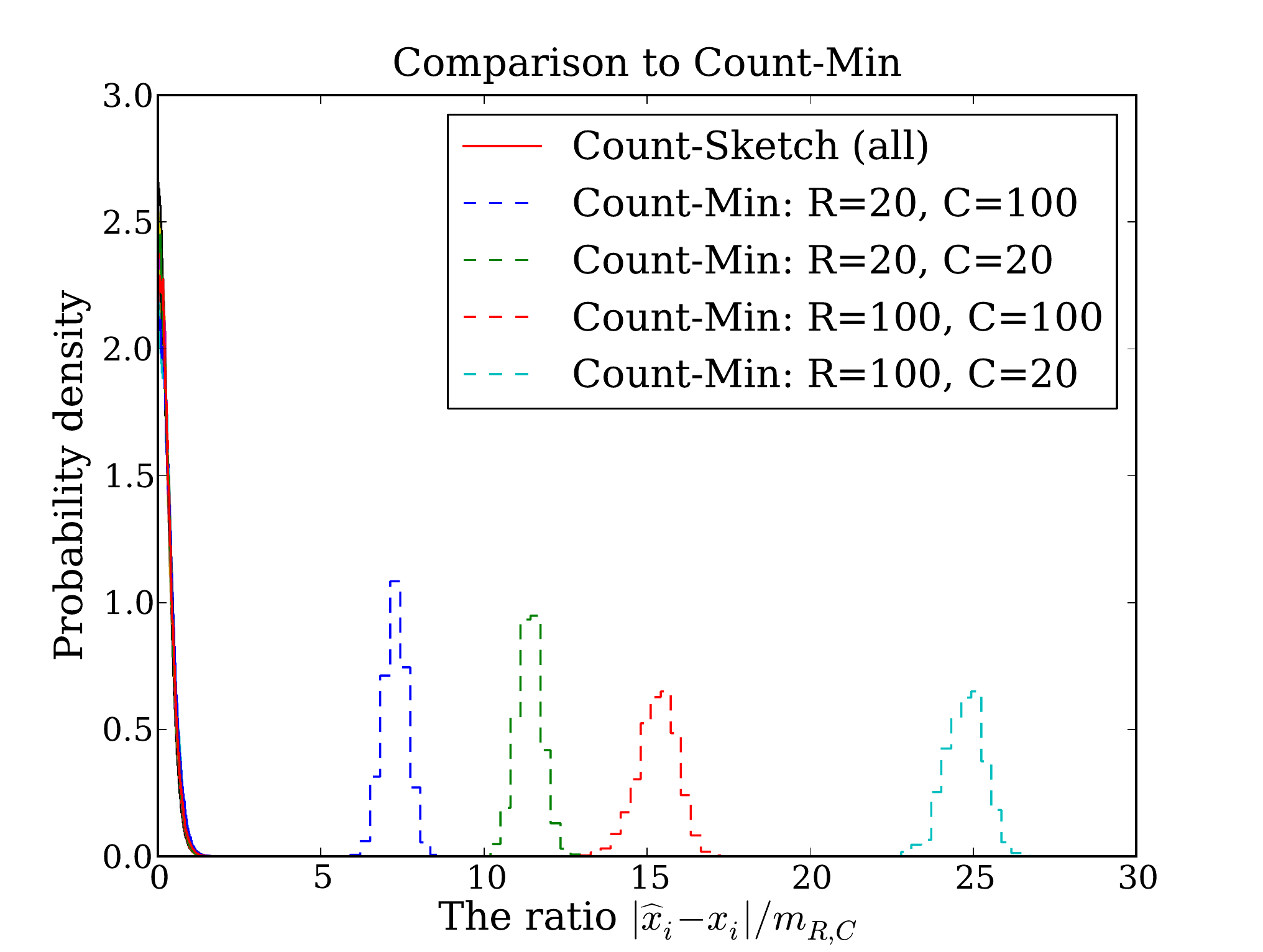}
  }
  \caption{Histograms of the point error $\PE$}
  \label{fig:PE}
\end{figure}
\item (Theorem~\ref{thm:l2}) After removing $n2^{-\Omega(R)}$
  probability mass, the top-$k$ estimation error has expectation
  \[
  \E[\KE] \simeq \sqrt{k}m_{R,C}.
  \]
  Furthermore, as $k$ increases, $\KE$ concentrates more strongly
  about its mean:
  \[
  \Pr[\KE > 2\E[\KE]] \lesssim 1/k.
  \]
\end{itemize}

Our results are presented in the form of a series of figures.

\begin{figure}[H]
  \centering

  \subfloat[Distribution of $\KE/(m_{R,C}\sqrt{k})$ for
  multiple $C$.  Once $C$ is large enough, the distribution is fairly
  static.]{
    \includegraphics[width=0.5\textwidth]{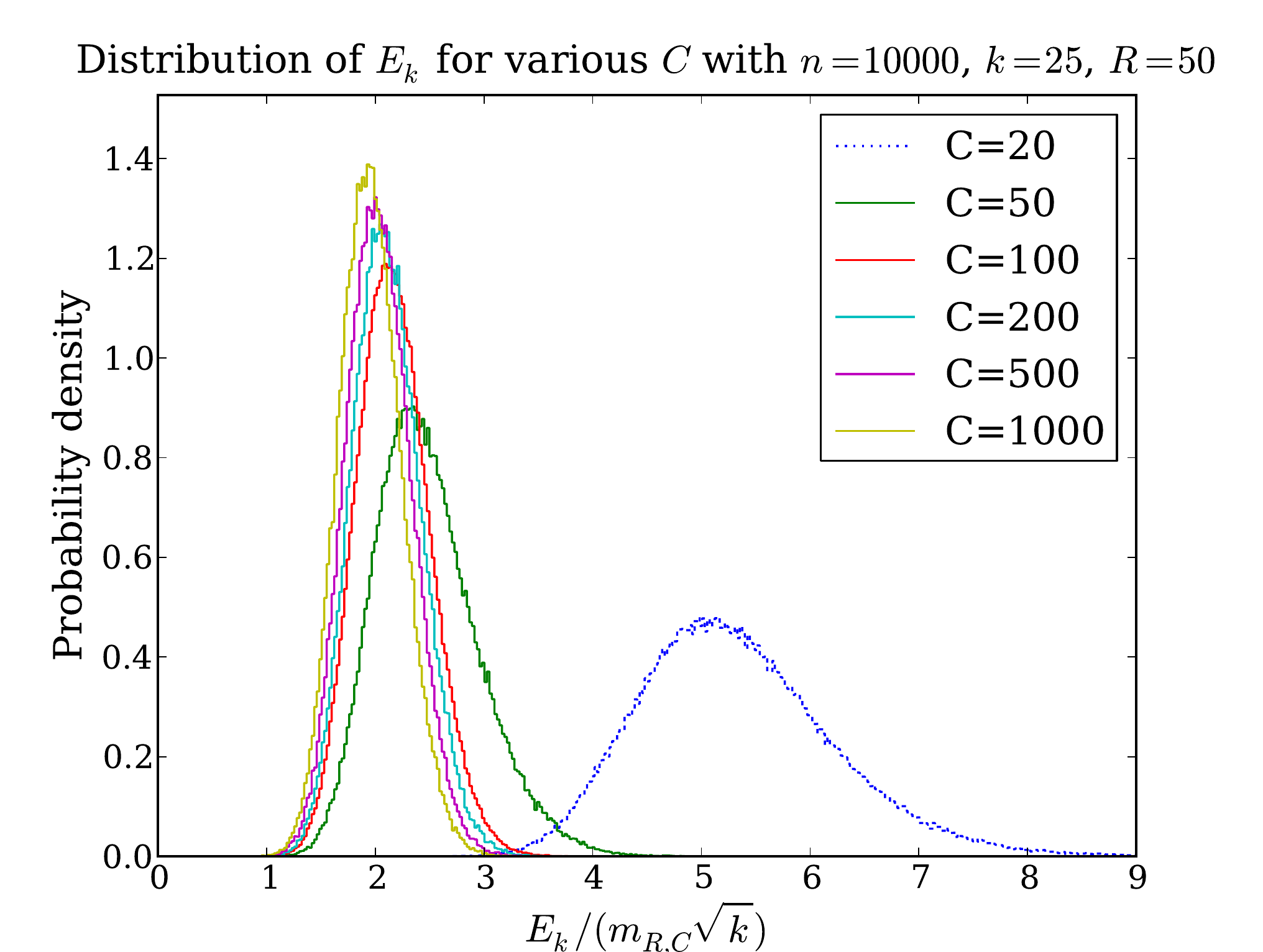}
  }
  \subfloat[Distribution of $\KE/(m_{R,C}\sqrt{k})$ for multiple $R$.
  Once $R$ is large enough, the distribution is fairly static.]{
    \includegraphics[width=0.5\textwidth]{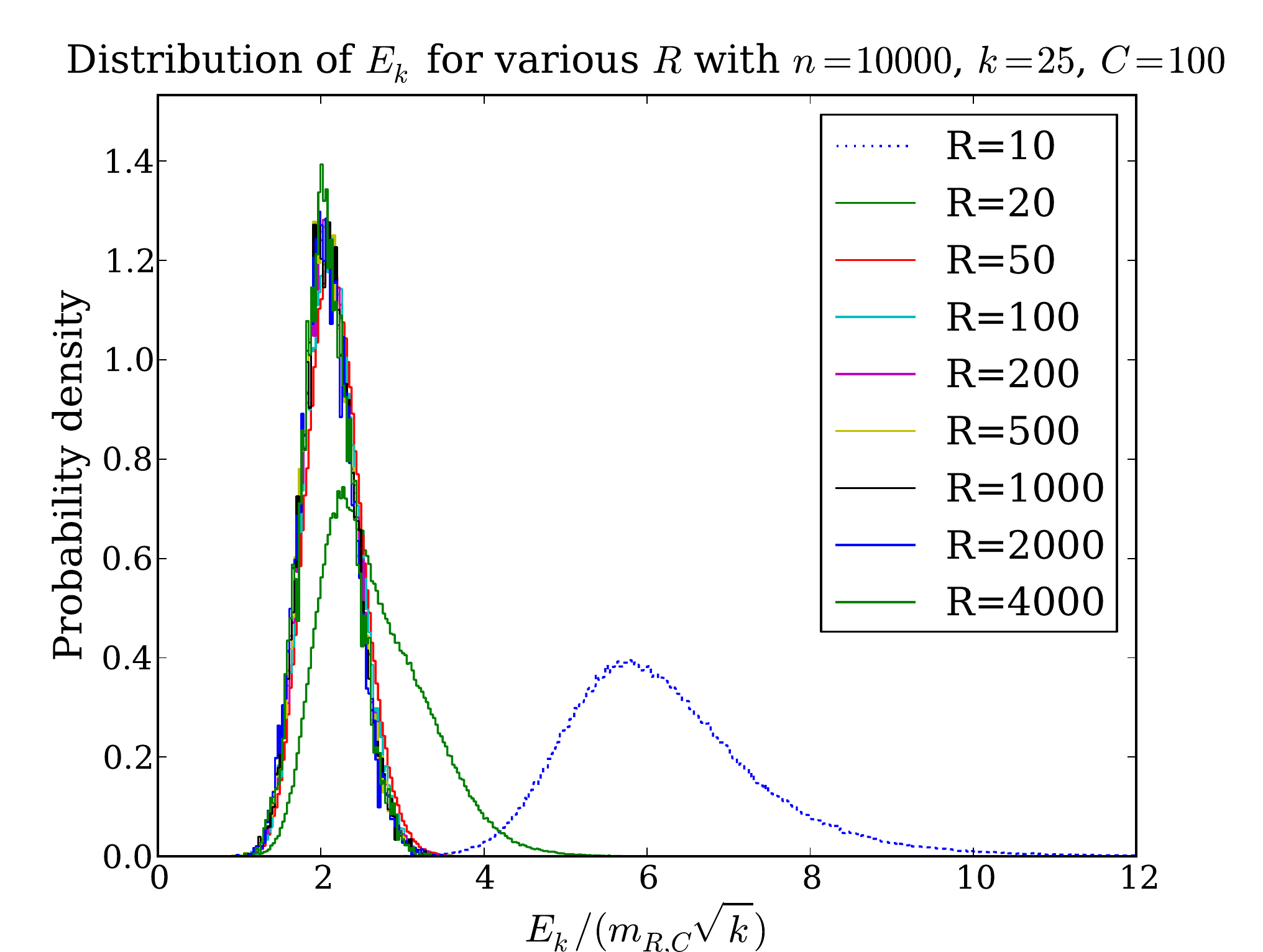}
  }\\
  \subfloat[{$\E[\KE]/(m_{R,C}\sqrt{k})$ as a function of $C$.  Above some
  threshold (depending on $R$), the mean remains fixed at a constant.}]{
    \includegraphics[width=0.5\textwidth]{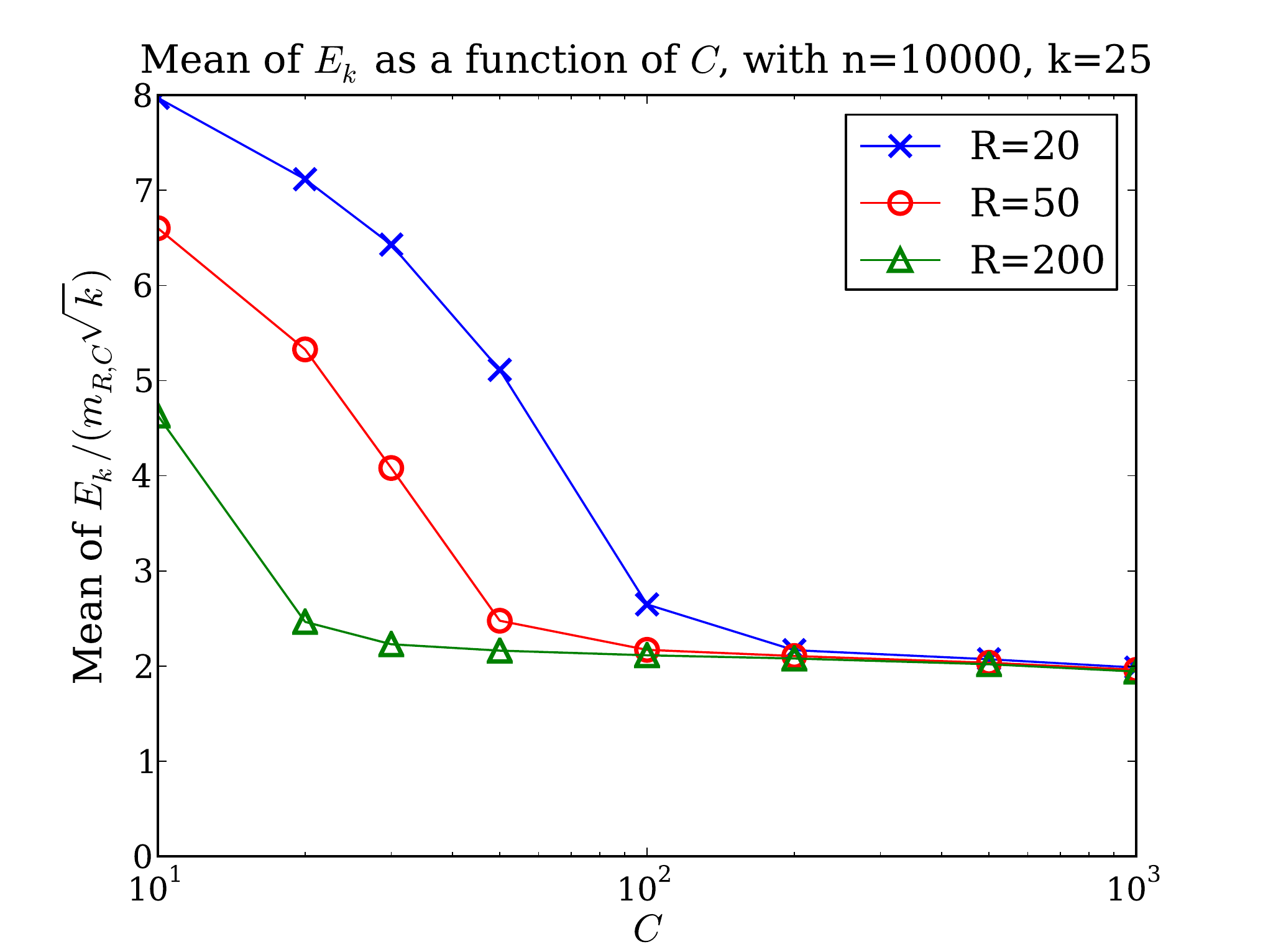}
  }
  \subfloat[{$\E[\KE]/(m_{R,C}\sqrt{k})$ as a function of $R$. Above some
  threshold (depending on $C$), the mean remains fixed at a constant.}]{
    \includegraphics[width=0.5\textwidth]{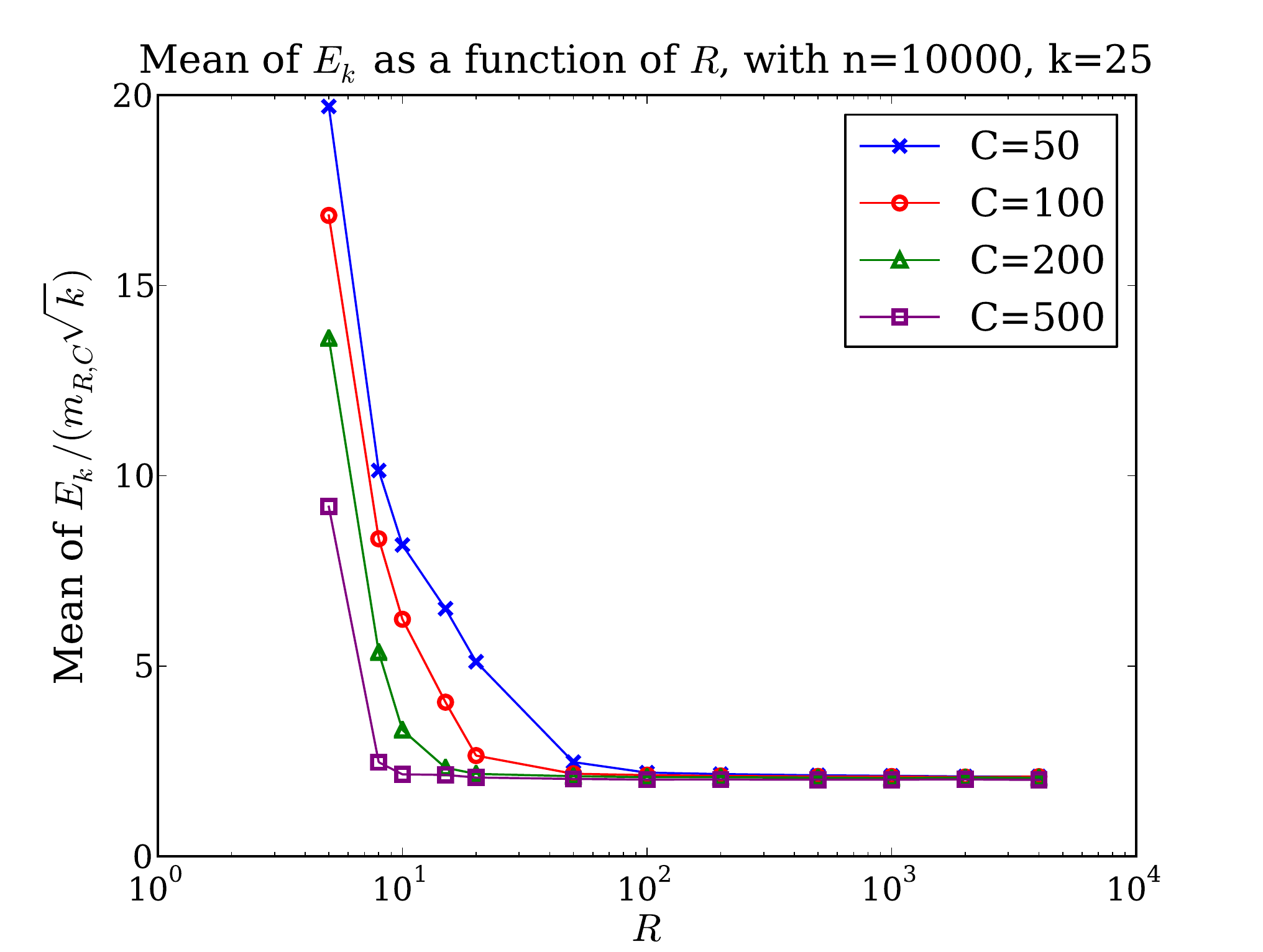}
  }
  \caption{Experimental results for the top-$k$ error $\KE$}
  \label{fig:KE}
\end{figure}
\begin{figure}[H]
  \centering

  \subfloat[Distribution of $\KE/(m_{R,C}\sqrt{k})$ for varying $k$.
  For small $k$, as $k$ increases the distribution becomes narrower
  and remains roughly in place.  For larger $k$, as $k$ increases the
  distribution instead shifts to the right.]{
    \includegraphics[width=0.5\textwidth]{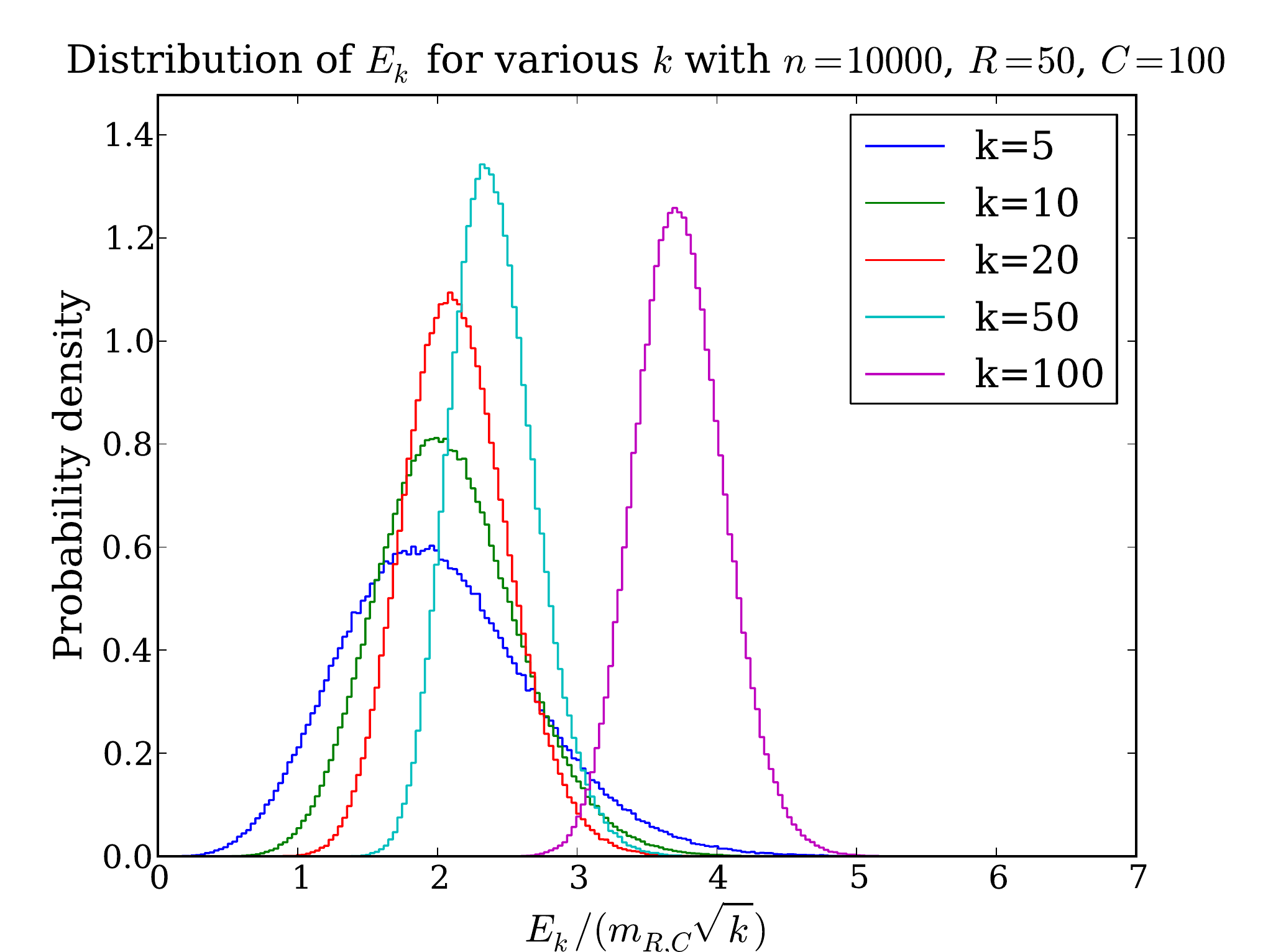}
  }
  \subfloat[The variance of {$\KE/\E[\KE]$}, as a
  distribution over $k$.  It appears to be approximated well by $0.6/k$.]{
    \includegraphics[width=0.5\textwidth]{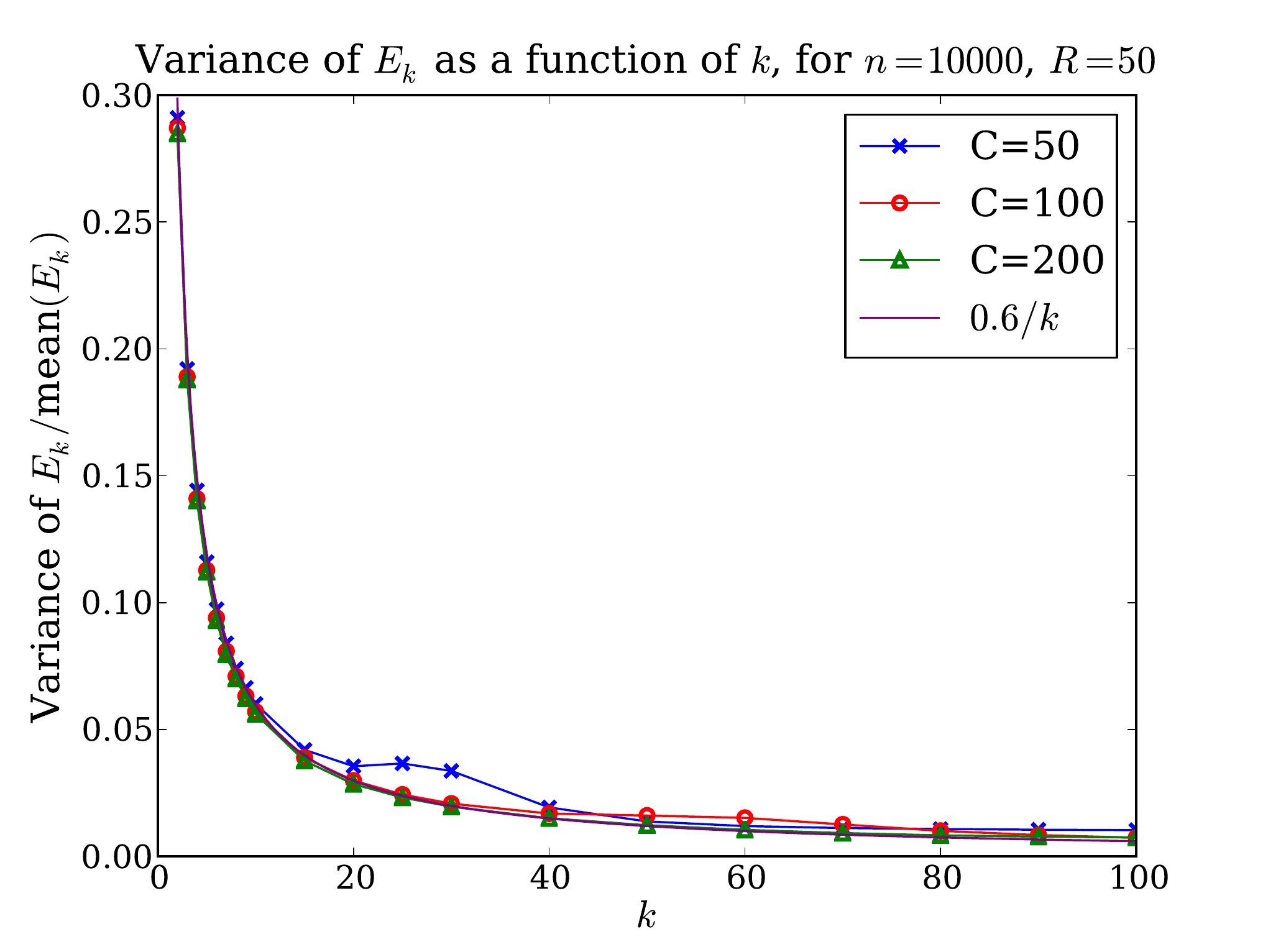}
  }
  \caption{Experimental results for the variation of the top-$k$ error $\KE$}
  \label{fig:KEvar}
\end{figure}

Figure~\ref{fig:PE}\subref{f:ped} shows the probability density
function of $\PE/m_{R,C}$ for $n=10^6$ and many different pairs $(R,
C)$.  We find that the PDFs all look fairly similar, and match a
Gaussian with constant standard deviation.  For comparison,
Figure~\ref{fig:PE}\subref{f:pecm} shows the equivalent error when
using the Count-Min sketch of Cormode and Muthukrishnan~\cite{CM04}.  We find that
Count-Min gives asymptotically higher error for the estimation of each
coordinate.

We study the distribution of $\KE$ in Figure~\ref{fig:KE}.
Figures~\ref{fig:KE}(a) and \ref{fig:KE}(b) give the probability
density functions of $\KE/(m_{R,C}\sqrt{k})$ for various $R$ and $C$,
respectively.  In both cases, we find that once $R$ and $C$ reach a
threshold, the distribution remains roughly constant---and has a
constant mean---as $R$ and $C$ increase beyond that
point. Figures~\ref{fig:KE}(c) and \ref{fig:KE}(d) show how $\E[\KE]$
changes as a function of $R$ and $C$, respectively.  As predicted, we
find that $\KE/(m_{R,C}\sqrt{k})$ has constant mean---so $\KE$ scales
as $m_{R,C}\sqrt{k}$---after $R$ and $C$ are sufficiently large.  The
threshold above which $\KE/(m_{R,C}\sqrt{k}) \approx 3$ allows some
trade-off between $R$ and $C$.  At $n = 10^4$ and $k = 25$, we observe
that $(R, C) = (26, 100) \approx (2\log_2 n, 4k)$ is above the
threshold.

In Figure~\ref{fig:KEvar}, we consider how well $\KE$ concentrates
about its mean as a function of $k$.  In \ref{fig:KEvar}(a), we plot
the PDF of $\KE/(m_{R,C}\sqrt{k})$ for various values of $k$.  We
observe that as $k$ increases, the distribution becomes more tightly
distributed about its mean.  However, once $k$ is large enough, our
chosen $(R, C)$ is no longer above the threshold for $(n, k)$, causing the
distribution of $\KE/(m_{R,C}\sqrt{k})$ to shift markedly to the right and
stop becoming more tightly distributed.  In \ref{fig:KEvar}(b), we plot the
variance of $\KE/\E[\KE]$, as a function of $k$.  We find that it is $\Theta(1/k)$, which gives
\[
\Pr[\KE > 2\E[\KE]] \lesssim 1/k.
\]
This is the analog of Theorem~\ref{thm:bettersets}.

\bibliographystyle{alpha}
\bibliography{paper}

\appendix

\section{Proofs for Concentration of Sets} \label{app:sets}

In this section our aim is to prove Proposition~\ref{prp:sets}.
Our approach is to study the pairwise correlations between
errors in coordinates.  We do this in two parts.
We first define a variant of Count-Sketch,
which we call \emph{tail-independent modified (TIM) Count-Sketch}.
In TIM Count-Sketch, the error coming from collisions with small elements
(the ``tail error") is replaced by independent, uniform noise.   We then focus on two
fixed coordinates and define a further variant,
\emph{fully-independent modified (FIM) Count-Sketch}.  In FIM Count-Sketch,
the errors for our two coordinates of interest are fully independent.
We define these variants in such a way that
unmodified, TIM, and FIM Count-Sketch can all be sampled with the
same randomness, so that they may be compared simultaneously.
We then bound the covariance between coordinate errors in
TIM Count-Sketch by bounding the difference between TIM and
FIM, and use the resulting bound on the error of sets in TIM Count-Sketch
to conclude a bound on unmodified Count-Sketch.

As in the statement of Proposition~\ref{prp:sets}, let $S \subset [n]$
be a subset of indices with $\abs{S} \le k$; this is the set on which we
will study concentration.  In the first part of the argument we will work
instead with $S' := S \cup [k]$, the set of interest together with the set of
heavy hitters.  Let $C = ck$ be the number of columns
in our sketch, for sufficiently large constant $c$.
Let $R \gtrsim \log k$ be the number of rows, and let $x$ be
the vector we are sketching.  Finally, define
$\mu = \norm{2}{x_{\overline{S'}}}/\sqrt{C} \le \norm{2}{\tail{x}{k}}/\sqrt{C}$.

\begin{observation} \label{obs:resamplebigger}
Let $X$ be a symmetric random variable and suppose $\eta,\sigma > 0$ are such that, for all $\epsilon \le 1$,
$\Pr[|X| < \epsilon \sigma] \ge \eta \epsilon$.
(For instance, if $X$ is a random variable to which Corollary~\ref{cor:gen} applies
and $\sigma^2 = \E[X^2]$, then one can check that $\eta=1/7$ suffices.)
Let $U$ be a symmetric random variable which is uniform on
$[-1,1]$ with probability $\eta$ and otherwise is $\pm \infty$.
Then $\Pr[|X| < \epsilon \sigma] \ge \Pr[\sigma |U| < \epsilon \sigma]$
for all $\epsilon$.  It follows that we can sample
$X$ and $U$ in such a way that they always have the
same sign and satisfy $|X| \le \sigma |U|$.
\end{observation}

\begin{lemma} \label{lma:maketailbigger}
As in a row of Count-Sketch, randomly assign to each
$i \in \overline{S'}$ a column $h(i)$ and a sign $s(i)$.
For each column $j$ let $T_j = \sum_{i \in \overline{S'} : h(i)=j} s(i) x_i$.
Let $L$ be a subset of $\ell \le 2k$ columns.  There exist
i.i.d.\ symmetric random variables $\{V_j : j \in L\}$ with
the following properties:
(i) $V_j$ is uniform on $[-\mu,\mu]$ with constant probability
and otherwise is $\pm \infty$;
(ii) $T_j$ and $V_j$ have the same sign; and
(iii) $\abs{T_j} \le \abs{V_j}$.
\end{lemma}
\begin{proof}
Let $p = 2\ell/C$, so that for each $i$ we have
$\Pr[h(i) \in L] = \tfrac{1}{2}p$.
Recalling that $C = ck$ for sufficiently large $c$, by
taking $c \ge 8$ we can ensure $p \le 1/2$.

Consider the following alternative procedure for sampling the random variables $\{T_j : j \in L\}$:
for each $i \in \overline{S'}$, decide with probability $p$ if $h(i) \in L$, and if so
(1) assign $h(i)$ uniformly to one of the columns $j$ in $L$ and
(2) add to $T_j$ a symmetric random variable which is $0$ with probability
$1/2$ and otherwise is $\pm x_i$.  (The variables $\{T_j : j \not\in L\}$ are not computed.)
In other words, for each $i \in \overline{S'}$ we double the probability of $h(i) \in L$,
but offset that by introducing a $1/2$ probability that $i$ contributes zero.
It is clear that this is equivalent to the original definition of $T_j$, so we henceforth work with it.

Now condition on the column assignments $h(i)$.  Having done so,
the variables $T_j$ are independent.  Moreover, each is a sum of independent,
symmetric random variables which are $0$ with probability $1/2$.
Thus Corollary \ref{cor:gen} applies.  Let $\sigma_j$ be the
standard deviation of $T_j$ and let $\{U_j : j \in L\}$ be a set of i.i.d.\ random
variables distributed like the variable $U$ in
Observation \ref{obs:resamplebigger}.  Then, by that observation,
we can sample $T_j$ so that $T_j$ and $U_j$ have the same
sign and $|T_j| \le \sigma_j |U_j|$.

Removing the conditioning on column assignments,
the dependence between columns manifests in the random variables
$\sigma_j^2 = \sum_{i \in \overline{S'} : h(i) = j} x_i^2/2$ for $j \in L$.  Consider
the following procedure for sampling these variances.
\begin{enumerate}
\item Let $q \in (0,1)$ be the solution
to the equation $(1-q)^\ell = 1-p$.
\item For each $i \in \overline{S'}$, determine preliminary column
assignments by, for each column $j \in L$, deciding independently with probability $q$ if $i$
is to be placed in column $j$.  These preliminary assignments may have repetitions
and may not assign $i$ to any column.
\item If a coordinate $i$ is assigned to just one column, let that be $h(i)$.
If it is assigned to more than one column, then randomly choose one of those columns
to be $h(i)$.  If it is assigned to no columns, then we set $h(i) \in \overline{L}$, i.e.,
we effectively just ignore $i$. 
\item Let $\sigma_j^2 = \sum_{i \in \overline{S'} : h(i) = j} x_i^2/2$.
\end{enumerate}
One can easily check that this is a valid way
of sampling.  The probability $q$ is the solution
to the equation $\log (1-p) = \ell \log (1-q)$;
since $\log (1-x) = \Theta(-x)$ for $\abs{x} \le 1/2$ it follows
that $q = \Theta(p/\ell) = \Theta(1/C)$.

Let $\wh\sigma_j^2$ be the variance that would have
been obtained from this procedure had we omitted
step 3, i.e., had we not corrected double-assignments.
Note that the random variables $\{\wh\sigma_j^2\}$ are
i.i.d.\ and they satisfy $\sigma_j^2 \le \wh\sigma_j^2$.
In particular, we have $|T_j| \le \wh\sigma_j |U_j|$.

The random variables $\{\wh\sigma_j^2 : j \in L\}$ have expected value
$\E[\wh\sigma_j^2] = (q/2) \sum_{i \in \overline{S'}} x_i^2 = \Theta(\mu^2)$.
Thus, for any $\epsilon \le 1$, by Markov's inequality we have
\[\Pr[\wh\sigma_j |U_j| \le \epsilon \mu] \ge \Pr\left[\wh\sigma^2_j < 4 \E[\wh\sigma_j^2]\right] \cdot \Pr\left[|U_j| \le \epsilon \cdot \frac{\mu}{2\sqrt{\E[\wh\sigma_j^2]}} \right] \gtrsim \epsilon.\]
Applying Observation~\ref{obs:resamplebigger} once more, we find
symmetric random variables $V_j$ which are uniform on $[-\mu,\mu]$
with constant probability, which have the same sign as $U_i$ (and thus
the same sign as $T_j$), and which always satisfy $|V_j| \ge
\wh\sigma_j |U_j| \ge |T_j|$.  These random variables have all the
desired properties.

Finally, we note that, as written, the random variables just constructed
depend on $\ell$, in that $q$ depends on $\ell$.  However, we can remove this
dependence by simply choosing the largest $q$ over all $\ell \le 2k$.  This gives
random variables which satisfy the same bounds but are agnostic about $\ell$.
\end{proof}

In a moment we will define the tail-independent modified (TIM) Count-Sketch.
The main point of TIM Count-Sketch is to replace the actual contributions of the ``tail'' coordinates
$\overline{[k]}$ with the independent, uniform random variables in
Lemma~\ref{lma:maketailbigger}.
There is one additional difference, though:
for later use, we invent a notion of ``ghost coordinates" for TIM Count-Sketch.
We do this for the following reason.  In TIM Count-Sketch, when there is a collision
between two coordinates in $S'$, we will assign them each a fixed error
instead of using the tail noise.  Later, when we modify TIM Count-Sketch to achieve
full pairwise independence, it will be convenient to have a larger probability of
using the fixed error than we get from just collisions between coordinates in $S'$.  
The right probability for our uses lies somewhere between that which you get from
considering collisions amongst $|S'|$ coordinates and that which you get from
considering collisions amongst $C+1$ coordinates.  To achieve
this intermediate probability, we fabricate $C+1-|S'|$ ``ghosts".
These are dummy coordinates whose only purpose is to (maybe) collide with
coordinates of $S'$ to force them to use the fixed error.  To allow us to tune
the probability of collision, each ghost may or may not be ``real", according to
i.i.d.\ Bernoulli random variables.  Thus the probability of a ghost colliding
with a fixed coordinate $i$ is the probability of that ghost being real times
the probability that it is assigned the same column as $i$.

We are now ready to give an actual definition.
Fix a bound $M > 0$ (later we will take $M = \mu$) and,
for each row, compute estimates as follows.
\begin{enumerate}
\item Assign signs $s(i)$ and columns $h(i)$ to the elements of $S'$.
\item Choose signs and columns for the elements of $\overline{S'}$
and, for the columns $j$ occupied by elements of $S'$, let
$T_j$ and $V_j$ be as in Lemma~\ref{lma:maketailbigger}.
\item Fabricate $C+1-|S'|$ ghost coordinates and, for each, decide
independently with probability $p_{\text{ghost}}$ if that ghost is real.
Random choose a column for each ghost that is real.
(The probability $p_{\text{ghost}}$ will be chosen later.)
\item For each coordinate $i \in S'$,
\begin{enumerate}
\item Let $H_i$ be the sum of $s(i) s(j) x_j$
over all $j \in S'$ with $j \ne i$ and $h(j) = h(i)$.
\item Let $\sigma \in \{\pm 1\}$ be the sign of $H_i + s(i) T_{h(i)}$
(which would be the error in unmodified Count-Sketch).
\item If $i$ is in the same column as $j$ for some $j \in S'$,
$j \ne i$ (i.e., if the sum defining $H_i$ is not empty) or
if $i$ is in the same column as a ghost, then return $x_i + \sigma \cdot M$
as the estimate.
\item Otherwise return $x_i + \sigma \cdot \min \{M, \abs{V_{h(i)}} \}$
as the estimate for $i$ in this row.
\end{enumerate}
\end{enumerate}
The final estimate for each coordinate $i \in S'$ is
the median of the estimates in each row.  (TIM Count-Sketch
only yields estimates for the coordinates in $S'$.)

The tail-independence modification can only worsen errors,
in the following sense.

\begin{observation} \label{obs:medianbigger}
Suppose $a = \{a_1,\dots,a_n\}$ and $b = \{b_1,\dots,b_n\}$ are sequences
such that, for each $i \in [n]$, $a_i$ and $b_i$ have the same
sign and satisfy either $|a_i| \le |b_i|$ or $|b_i| \ge M$.
Then $\median a_i$ and $\median b_i$ have the same sign and satisfy
either $\abs{\median a_i} \le \abs{\median b_i}$ or
$\abs{\median b_i} \ge M$.
\end{observation}

\define{lma:tics}{Lemma}{Let $\wh{x}_{\text{tim}}$ be the estimate of $x$ using
TIM Count-Sketch and let $\wh{x}_{\text{um}}$ be the estimate using unmodified
Count-Sketch.  For any subset $A \subset (-M,M)^{|S|}$ which is convex
and symmetric in each coordinate,
\[\Pr[(\wh{x}_{\text{um}} - x)_{S} \in A] \ge \Pr[(\wh{x}_{\text{tim}} - x)_{S} \in A].\]}
\declare{lma:tics}
\begin{proof}
Note that unmodified Count-Sketch can be run simultaneously
with TIM Count-Sketch, using the same randomness.
Consider a fixed row and a fixed coordinate $i \in S$.
Keeping the notation above, the error $E_{\text{tim}}$ arising from
TIM Count-Sketch is $\sigma \cdot M$ if there is a
collision or $\sigma \cdot \min\{M,|V_{h(i)}|\}$ if not.  The error
$E_{\text{um}}$ arising from unmodified Count-Sketch is $H_i + s(i) T_{h(i)}$.
Clearly $E_{\text{tim}}$ and $E_{\text{um}}$ always
have the same sign and satisfy either
$\abs{E_{\text{um}}} \le \abs{E_{\text{tim}}}$ or
$\abs{E_{\text{tim}}} \ge M$.
The final errors in the estimates are medians of these row errors.
Thus, by Observation \ref{obs:medianbigger},
\[\abs{(\wh{x}_{\text{um}} - x)_i} \le \abs{(\wh{x}_{\text{tim}} - x)_i}\]
or else the TIM error is at least $M$.
In other words, given this method of sampling, whenever
the error for TIM Count-Sketch is less than $M$ we know
that the error for unmodified Count-Sketch is no bigger
than the error for TIM Count-Sketch.  This clearly proves
what we wanted.
\end{proof}

Now that we have arranged for independence of the tail
contributions, the only remaining dependence arises from collisions
amongst the elements of $S'$.  Fix two coordinates $i_1,i_2 \in S'$;
we will bound the correlation between the errors in
these two coordinates.  Analogously to our analysis of
$\sigma_j$ in Lemma~\ref{lma:maketailbigger}, we can highlight
the dependence by first pretending collisions are independent
and then correcting double-collisions.  More precisely,
consider the following alternative mechanism for determining
collisions.
\begin{enumerate}
\item  Let $\mathfrak{p} \colon [0,\tfrac{1}{2}] \rightarrow [0,1]$
denote the inverse of the monotone-increasing function $p \mapsto p/(1+p)$.
\item For starters, declare that $i_1$ and $i_2$ do not collide.  (This may change later
in the procedure.)
\item For each element $j$ of $S' \setminus \{i_1,i_2\}$ (resp., each ghost)
and each of $i=i_1,i_2$, independently decide with
probability $\mathfrak{p}(1/C)$ (resp., $\mathfrak{p}(p_{\text{ghost}}/C)$)
if $j$ collides with $i$.  Note that, because these decisions are independent,
there may well be double-collisions at this stage.
\item If any ghost or coordinate in $S'$ collides with both
$i_1$ and $i_2$, then resample everything according to
the correct distribution, conditioned on $i_1$ and $i_2$ colliding.
\end{enumerate}

Using this procedure, the event that $i_1$ and $i_2$ end up not colliding is
the same as the event that step 3 produced no double-collisions.  The
probability of this is
\[p_{\text{nc}}(p_{\text{ghost}}) := \Pr[\text{$i_1,i_2$ do not collide}] = (1-\mathfrak{p}(1/C)^2)^{|S'|-2}(1-\mathfrak{p}(p_{\text{ghost}}/C)^2)^{C+1-|S'|}.\]
This is a monotone-decreasing function of $p_{\text{ghost}}$.  Noting that
$\mathfrak{p}(1/C) = 1/(C-1)$, we see that
\begin{equation*}
p_{\text{nc}}(0) = \left ( 1 - \frac{1}{(C-1)^2} \right )^{|S'|-2} \ge \left ( 1 - \frac{1}{(C-1)^2} \right )^{2k} = 1 - \frac{1}{C} \cdot \Theta \left ( \frac{k}{C} \right )
\end{equation*}
and
\begin{equation*}
p_{\text{nc}}(1) = \left(1-\frac{1}{(C-1)^2}\right)^{C-1}.
\end{equation*}
By taking $C$ to be a suitably large multiple of $k$ we can
arrange for $p_{\text{nc}}(0)$ to be at least $1-1/C$.
By a simple calculus exercise, $p_{\text{nc}}(1) \le 1-1/C$.
Thus there is a unique $q \in [0,1]$ such that
$p_{\text{nc}}(q) = 1-1/C$.  We now and henceforth set $p_{\text{ghost}}$ to
this value of $q$.

This is supposed to be an alternative, but equivalent,
method for determining collisions.  Before continuing, let
us check that it is indeed equivalent.  Our choice of
$p_{\text{ghost}}$ guaranteed that the new mechanism has the
right probability of $i_1$ and $i_2$ colliding; moreover,
when they do collide, we explicitly sample according to
the correct distribution.  Thus, to demonstrate equivalence,
we need only to consider the case when $i_1$ and $i_2$ do not
collide.  Condition on this event and consider the
(conditional) probability of some other coordinate
$j \in S' \setminus \{i_1,i_2\}$ colliding
with $i_1$.  Using the original sampling mechanism, this probability
is $1/C$.  Using our alternative sampling mechanism, the probability is
\begin{align*}
& \Pr[\text{$j$ collides with $i_1$} \mid \text{$i_1,i_2$ do not collide}] \\
= &\Pr[\text{$j$ collides with $i_1$ in step $3$} \mid \text{there are no double-collisions in step 3}].
\end{align*}
The event that there is no double-collision is the intersection
of independent events for each element of $S' \setminus \{i_1,i_2\}$
and for each ghost.  Of these constituent events, only one is relevant
to the conditional probability we want to compute: the event that $j$
does not double-collide.  Thus our probability is
\begin{align*}
& \Pr[\text{$j$ collides with $i_1$ in step $3$} \mid \text{$j$ does not collide with both $i_1$ and $i_2$ in step 3}] \\
= & \frac{\mathfrak{p}(1/C) - \mathfrak{p}(1/C)^2}{1 - \mathfrak{p}(1/C)^2} = \frac{\mathfrak{p}(1/C)}{1 + \mathfrak{p}(1/C)} = \frac{1}{C},
\end{align*}
which is what we wanted.
One can similarly check that the probability of a ghost
collision is correct.  Thus, as claimed, our new scheme is
a valid way to sample the collision events.

We can now define our last variant of Count-Sketch,
fully-independent modified (FIM) Count-Sketch.  This only produces estimates for the two
coordinates $i_1$ and $i_2$.   For each row, the FIM
estimate is computed as follows.  We always specify
that $i_1$ and $i_2$ will not collide.
To determine which elements of $S' \setminus \{i_1,i_2\}$
and which ghosts collide with $i_1$ and $i_2$, we use
(1--3) of the ``alternative mechanism'' above.  We
omit step 4, so that a given coordinate may collide
with both $i_1$ and $i_2$.  Then, for each $i \in \{i_1,i_2\}$
and for each coordinate colliding with $i$, we choose a
random sign.  (In particular, if some coordinate is supposed to
collide with both $i_1$ and $i_2$, then it is associated with
two different, independent random signs.)  Using
these collision and sign data, we proceed
as in (4a--4d) of the description of TIM
Count-Sketch to get a row estimate.  As always, the final estimate is the median
of the row estimates.

In each row, the estimates for $i_1$ and $i_2$ produced
by FIM Count-Sketch are independent; thus the final
estimates are also independent.

FIM and, to a lesser extent, TIM Count-Sketch are quite a bit different
from unmodified Count-Sketch.
However, for a single coordinate they preserve many of the salient features.
In particular, all of the properties used in the proof of Theorem~\ref{thm:main}
still hold: rows are independent, the errors in each row are symmetric, and
in each row we have (i) with constant probability, there is no collision with
$[k]$ and (ii) the error arising from collisions with $\overline{[k]}$ satisfies
the bound of Corollary~\ref{cor:gen}.  Thus, with the same proof as Theorem~\ref{thm:main},
we have the following bounds.

\begin{lemma} \label{lma:fimbound}
Fix $i \in \{i_1,i_2\}$ and consider the estimates $(\wh{x}_{\text{tim}})_i$
and $(\wh{x}_{\text{fim}})_i$ of $x_i$ from
TIM and FIM Count-Sketch, respectively.  For any $t \le R$ we have
 \[
  \Pr\left[((\wh{x}_{\text{tim}})_i - x_i)^2 > \frac{t}{R} \cdot \frac{\norm{2}{\tail{x}{k}}^2}{k}\right] < 2e^{-\Omega(t)}
\quad\text{and}\quad
  \Pr\left[((\wh{x}_{\text{fim}})_i - x_i)^2 > \frac{t}{R} \cdot \frac{\norm{2}{\tail{x}{k}}^2}{k}\right] < 2e^{-\Omega(t)}.
\]
\end{lemma}

\define{cor:timfimvariance}{Corollary}{Suppose $M = \Theta(\mu)$.  Then

\centering$\renewcommand*{\arraystretch}{2}\begin{array}{lclc}
\E[((\wh{x}_{\text{tim}})_i - x_i)^2] \lesssim \dfrac{\mu^2}{R} & , & \E[((\wh{x}_{\text{tim}})_i - x_i)^4] \lesssim \dfrac{\mu^4}{R^2} & , \\
\E[((\wh{x}_{\text{fim}})_i - x_i)^2] \lesssim \dfrac{\mu^2}{R} & \text{, and} & \E[((\wh{x}_{\text{fim}})_i - x_i)^4] \lesssim \dfrac{\mu^4}{R^2} & .
\end{array}$}
\declare{cor:timfimvariance}

We can recover TIM Count-Sketch from FIM Count-Sketch by
resampling some of the rows.  More specifically, we
take any row in which some ghost or some element of $S'$
collides with both $i_1$ and $i_2$ and resample that
row, conditioning on $i_1$ and $i_2$ colliding with
each other.  The errors for both $i_1$ and $i_2$ in
any such row are $M$ in magnitude both before and
after resampling.  Moreover, for a fixed coordinate
$i \in \{i_1,i_2\}$ the signs of the errors before and
after resampling are independent and uniform.  (Note:
the signs of the errors for $i_1$ and $i_2$ after resampling
are not necessarily independent of each other.  We assume nothing about
their dependence in our argument.)

Let $\wh{x}_{\text{fim}}$ be the estimate computed by
FIM Count-Sketch and let $\wh{x}_{\text{tim}}$ be the
result of TIM Count-Sketch, recovered from FIM
Count-Sketch as above.  Let $\widetilde R$ be the number of
rows that have to be resampled to recover TIM Count-Sketch from FIM Count-Sketch.
Focusing for the moment on $i_1$, define the errors
\[E_{\text{fim}} = (\wh{x}_{\text{fim}} - x)_{i_1}\ \ \ \text{and}\ \ \ E_{\text{tim}} = (\wh{x}_{\text{tim}} - x)_{i_1}.\]
We expect $\widetilde R$ to be reasonably small, and thus we expect the change in moving from
FIM to TIM to be small.  In particular, we expect $E_{\text{fim}}$ and $E_{\text{tim}}$
to be close.  More specifically, we shall prove the following proposition.

\define{prp:fimtim}{Proposition}{Suppose $M = \Theta(\mu)$.
Then $\displaystyle \E[(E_{\text{tim}}^2 - E_{\text{fim}}^2)^2] \lesssim \frac{M^4}{R^2 k^{1/7}}$.}
\declare{prp:fimtim}

Since we just want to bound the expected value of a bounded random variable, small
probability events can be ignored.  Indeed, in general, if $X$ is a
random variable bounded by $B$ and $\mathcal{E}$ is an event with probability $p$, then
\[\E[X] = p \E[X \mid \mathcal{E}] + (1 - p) \E[X \mid \overline{\mathcal{E}}] \le p B + \E[X \mid \overline{\mathcal{E}}].\]
Thus if we aim to prove $\E[X] \lesssim Y$ and we know $p B \lesssim Y$, then it suffices to prove
$\E[X \mid \overline{\mathcal{E}}] \lesssim Y$.  Specializing to our case, because
$\abs{E_{\text{fim}}},\abs{E_{\text{tim}}} \le M$ always holds, this demonstrates
that in order to prove Proposition~\ref{prp:fimtim}, we may condition away
from events that occur with probability $O(R^{-2} k^{-1/7})$.  We will refer to such
events as ``ignorable" and, as the name indicates, freely ignore them.

\begin{lemma} \label{lma:fimtimlargek}
If $k \ge R^{7/2}$ then Proposition~\ref{prp:fimtim} holds.
\end{lemma}
\begin{proof}
The probability that a single row is not resampled is $p_{\text{nc}} = 1 - 1/C$, and so
\[\Pr[\widetilde R = 0] = \left ( 1 - \frac{1}{C} \right )^R = \exp(R \log (1 - 1/C)) = \exp(-R \Theta(1/C)) = 1 - \Theta \left ( \frac{R}{C} \right ).\]
Because $C = \Omega(k)$ and $k \ge R^{7/2}$, $R/C \lesssim R/k \lesssim R^{-2} k^{-1/7}$.
Thus $\widetilde R > 0$ is an ignorable event.  But of course if $\widetilde R = 0$ then
$\wh{x}_{\text{fim}} = \wh{x}_{\text{tim}}$.  This gives the desired bound.
\end{proof}

\begin{observation} \label{obs:whatisignorable}
Let $a,b,c > 0$ be arbitrary.  Any event that occurs with probability
$\exp(-\Omega(R^b k^c))$ is ignorable.  Moreover, if we suppose that
$k \le R^a$, then any event that occurs with probability $\exp(-\Omega(R^b))$
is ignorable.
\end{observation}

\begin{lemma} \label{lma:fimtimsmallk1}
Suppose $k \le R^{7/2}$.  Then, except for ignorable events, $\widetilde R < R/k^{1/5}$.
\end{lemma}
\begin{proof}
Each row is resampled with probability $1/C$, and the resampling events for the $R$ rows
are independent.  Thus, by the Chernoff bound $\Pr[X \ge (1+\epsilon)\E[X]] \le \exp(-\E[X]\epsilon^2/(2+\epsilon))$
(applied with $\epsilon = C/k^{1/5}$, which in particular can be taken to be arbitrarily large),
\[\Pr \left[ \widetilde R \ge \frac{R}{k^{1/5}} \right] \le \exp(-\Omega(R/k^{1/5})) = \exp(-\Omega(R^{3/10})).\]
This is ignorable by Observation~\ref{obs:whatisignorable}.
\end{proof}

Let $\widetilde R_{+},\widetilde R_{0},\widetilde R_{-}$ be the number of rows resampled
in which the error increases from $-M$ to $M$, stays the same, and decreases with $M$ to $-M$,
respectively.  Thus $\widetilde R_{+} + \widetilde R_0 + \widetilde R_{-} = \widetilde R$.
The net effect of resampling is measured by $\widehat R := \widetilde R_{+} - \widetilde R_{-}$.

\begin{lemma} \label{lma:binomial}
Fix a positive integer $r$.  If
$X$ is an $(N,p)$ binomial random variable with $Np \ge 1$, then
$\E[(X-Np)^r] \lesssim (Np)^{r/2}$ (where the implied constant depends on $r$).
\end{lemma}
\begin{proof}
Let $Y$ be a random variable with
$\Pr[Y=1-p] = p$ and $\Pr[Y=-p]=1-p$, and let
$X-Np = Y_1+\cdots+Y_N$ where
$Y_1,\dots,Y_N$ are i.i.d.\ with the same distribution as $Y$.  Then
\[\E[(X-Np)^r] = \E[(Y_1+\cdots+Y_N)^r] = \sum_{m=1}^{\min \{r,N\}} \sum_{\substack{e_1 \ge \cdots \ge e_m \ge 1\\\sum e_i = r}} C_{e_1,\dots,e_m} \E[Y^{e_1}]\cdots \E[Y^{e_m}],\]
where the coefficient $C_{e_1,\dots,e_m}$ counts the number of ways to choose
an $r$-tuple of elements of an $N$-element set such that the most common element
occurs $e_1$ times, the next most common element occurs $e_2$ times, and so on.

We have $\E[Y] = 0$, so the terms with $e_i = 1$ (for any $i$) all vanish.  Since $\sum_{i=1}^m e_i = r$,
this leaves just the terms with $m \le r/2$.  Now for any $e \ge 2$, we have
\[\E[Y^e] = p(1-p)^e + (1-p)(-p)^e \le p(1-p)^2 + (1-p)p^2 = p(1-p) \le p,\]
and so
\[\E[(X-Np)^r] \le \sum_{m=1}^{\min \{r/2,N\}} \sum_{\substack{e_1 \ge \cdots \ge e_m \ge 2\\\sum e_i = r}} C_{e_1,\dots,e_m} p^m.\]

Now we can compute $C_{e_1,\dots,e_m}$ by first choosing which $m$ of the $N$ elements occur
and then choosing how to arrange them.  The number of choices for the former is clearly
bounded by $N^m$ and the number of choices for the latter is bounded by a function of $r$ only.
Thus $C_{e_1,\dots,e_m} \lesssim N^m$.  Moreover, the number of terms in the
second summation is bounded by a function of $r$ only.  This leaves us with
\[\E[(X-Np)^r] \lesssim \sum_{m=1}^{\min \{r/2,N\}} N^m p^m.\]
Given that $Np \ge 1$, the last term in the summation is dominant, giving
the desired bound.
\end{proof}

\begin{lemma} \label{lma:fimtimsmallk2}
Suppose $k \le R^{7/2}$.  Conditioning away from ignorable events, $\E[\widehat{R}^4] \lesssim R^2/k^{2/5}$.
\end{lemma}
\begin{proof}
By reducing $\widetilde R$, we may assume without loss of generality that
$\widehat R_0 = 0$.  Thus $\widehat R$ is just the value of a discrete random walk
of length $\widetilde R$.  By Lemma~\ref{lma:binomial} the fourth
moment of such a random walk is $O(\widetilde R^2)$.  Combining
this with Lemma~\ref{lma:fimtimsmallk1} gives the desired bound.
\end{proof}

\begin{lemma} \label{lma:fimtimsmallk3}
Suppose $k \le R^{7/2}$ and $M=\Theta(\mu)$.  Conditioning away from ignorable events
and then conditioning on the value of $\widehat{R}$,
$\E[(E_{\text{tim}} - E_{\text{fim}})^4] \lesssim (M \widehat{R}/R)^4$.
\end{lemma}
\begin{proof*}
First, suppose $\widehat{R}=0$.  In this case $E_{\text{fim}}=E_{\text{tim}}$,
so the desired bound holds because both sides are $0$.  Thus we may suppose $\widehat{R} \ne 0$;
moreover, by symmetry we may assume $\widehat{R} > 0$.  In this case we have
$E_{\text{tim}} \in [E_{\text{fim}},M]$.

By Lemma~\ref{lma:fimbound} we see that
$\Pr[|E_{\text{fim}}| > M/2] < e^{-\Omega(R)}$; thus,
by Observation~\ref{obs:whatisignorable}, the corresponding
event is ignorable.  For the remainder of the proof, assume that it does not happen.

Condition for a moment on both the value of $E_{\text{fim}}$
and on the set of $R/2$ rows in which the error for FIM Count-Sketch
is above the median.  Pick one such row $r$ and consider its error $E_r$.
Before the conditioning, using the assumption $M = \Theta(\mu)$,
the distribution for $E_r$ consisted of atoms at $\pm M$ and
$\Omega(1)$ probability of being uniform in $[-M,M]$.  The net effect
of our conditioning is to simply condition on $E_r \ge E_{\text{fim}}$.
(By conditioning first on $r$ having above-median error, we removed
the nontrivial dependence between $E_{\text{fim}}$ and $E_r$.)
In particular, with $\Omega(1)$ probability, $E_r \in [E_{\text{fim}},M)$.
(And, when this occurs, $E_r$ is uniform in that interval.)
Applying a Chernoff bound to the $R/2$
such rows, we see that, with probability $1-e^{-\Omega(R)}$, there are
$\Omega(R)$ rows in which the error lies in $[E_{\text{fim}},M)$.
Since the failure probability is ignorable, we henceforth
assume that this holds.

(Note that the ignorable events we just conditioned away 
influence $\widehat{R}$.  We conditioned on them first to remove that dependence.)

Fix $t > 0$ and consider the event $\mathcal{E}_t := \{E_{\text{tim}} > E_{\text{fim}} + t M \widehat{R}/R\}$.
Since the difference between TIM and FIM Count-Sketch is just replacing $\widehat{R}$
rows with error $-M$ by rows with error $+M$, this event is equivalent to FIM Count-Sketch
having fewer than $\widehat{R}$ rows in which the error is in the interval
$I = [E_{\text{fim}},E_{\text{fim}}+tM\widehat{R}/R]$.
Given the above discussion, there are $\Omega(R)$ rows in which the probability of
the error lying in that interval is $\Omega(t \widehat{R}/R)$.  More explicitly, for some
constants $a$ and $b$, there are $N := aR$ rows in which the probability of the error lying
in $I$ is (at least) $p := bt \widehat{R}/R$.  Let $X$ be the total number of such rows.
Then $X$ is a $(N,p)$ binomial random variable.  For $t > t_0 := 1/(ab)$
(i.e., for all $t$ larger than a sufficiently large constant), we have $abt > 1$.
In particular this implies $Np > \widehat{R} \ge 1$, whence the $10$th moment of
$X$ is $O((Np)^5)$ by Lemma~\ref{lma:binomial}.  It also implies that
$\widehat{R} < \E[X] = abt \widehat{R}$.
Thus, by Markov's inequality, \`a la Chebyshev's inequality,
\[\Pr[\mathcal{E}_t] = \Pr[X < \widehat{R}] \le \frac{\E[(X-abt\widehat{R})^{10}]}{\widehat{R}^{10}(abt-1)^{10}} \lesssim \frac{(abt \widehat{R})^5}{\widehat{R}^{10}(abt-1)^{10}} \le \frac{(abt)^5}{(abt - 1)^{10}}.\]
For, say, $t \ge 2t_0$, this is asymptotically $\lesssim t^{-5}$.

Using integration by parts,
\[\E[(E_{\text{tim}} - E_{\text{fim}})^4] = \left ( \frac{M \widehat{R}}{R} \right )^4 \cdot \int_{0}^\infty 4t^3 \cdot \Pr[\mathcal{E}_t]\ \text{d}t.\]
Writing
\[\int_{0}^\infty 4t^3 \cdot \Pr[\mathcal{E}_t]\ \text{d}t = \int_{0}^{2t_0} 4t^3 \cdot \Pr[\mathcal{E}_t]\ \text{d}t + \int_{2t_0}^\infty 4t^3 \cdot \Pr[\mathcal{E}_t]\ \text{d}t,\]
the first term is $O(1)$ because $t_0$ is a constant.
The second is $\lesssim \int_{2t_0}^{\infty} t^{-2}\ \text{d}t$, which is
also $O(1)$.  Thus, as desired,
\[\E[(E_{\text{tim}} - E_{\text{fim}})^4] \lesssim \left ( \frac{M \widehat{R}}{R} \right )^4. \tag*{\qedsymbol}\]
\end{proof*}

We are now ready to complete the

\begin{proof}[Proof of Proposition~\ref{prp:fimtim}]
After Lemma~\ref{lma:fimtimlargek} it only remains to handle the case $k \le R^{7/2}$.
Combining Lemmas \ref{lma:fimtimsmallk2} and \ref{lma:fimtimsmallk3}, we have
$\E[(E_{\text{tim}}-E_{\text{fim}})^4] \lesssim M^4/(R^2 k^{2/5})$.  Applying
Lemma~\ref{lma:fimbound} and the union bound, we see that
$\E[(E_{\text{tim}}+E_{\text{fim}})^4] \lesssim M^4/R^2$.  Thus,
\begin{align*}
\E[(E_{\text{tim}}^2 - E_{\text{fim}}^2)^2] &= \E[(E_{\text{tim}}-E_{\text{fim}})^2(E_{\text{tim}}+E_{\text{fim}})^2] \\
&\le \sqrt{\E[(E_{\text{tim}}-E_{\text{fim}})^4]\E[(E_{\text{tim}}+E_{\text{fim}})^4]} \\
&\lesssim \frac{M^4}{R^2 k^{1/5}},
\end{align*}
which gives the desired bound a fortiori.
\end{proof}

Using these lemmata, we can finally prove the desired result.

\state{prp:sets}
\begin{proof}
Throughout we condition on the error for each coordinate in $S$
being less than $\mu$.  By Theorem~\ref{thm:main} and a union bound,
this occurs with probability $1 - |S| e^{-\Omega(R)} = 1 - k^{-\Omega(1)}$,
so this conditioning can be absorbed into the final bound.

In addition to unmodified Count-Sketch, consider running TIM Count-Sketch
with $M = \mu$.  Define the convex, coordinate-wise symmetric set
\[A = \left \{ v \in (-\mu,\mu)^{|S|} : \norm{2}{v}^2 \le t \cdot \abs{S} \cdot \frac{1}{R} 
    \cdot \frac{\norm{2}{\tail{x}{k}}^2}{k} \right \}.\]
Applying Lemma~\ref{lma:tics} with this set $A$ shows that
unmodified Count-Sketch has at least as high of a probability
of its error lying in $A$ as TIM Count-Sketch.  Thus it suffices
to prove the asserted probability bound
for the TIM estimate $\wh{x}_{\text{tim}}$.

The error we are studying is $\norm{2}{(\wh{x}_{\text{tim}} - x)_S}^2 = \sum_{i \in S} E_{\text{tim},i}^2$,
where $E_{\text{tim},i} := (\wh{x}_{\text{tim}} - x)_i$.  By Corollary~\ref{cor:timfimvariance},
$\E[E_{\text{tim},i}^2] \le a\mu^2/R$ for some constant $a>0$.  Now, by Chebyshev's inequality, we have 
\begin{align}
\nonumber \Pr \left [ \norm{2}{(\wh{x}_{\text{tim}} - x)_S}^2 > (a+1) \cdot \frac{|S|\mu^2}{R} \right ] &\le \Pr \left [ \norm{2}{(\wh{x}_{\text{tim}} - x)_S}^2 - \E[\norm{2}{(\wh{x}_{\text{tim}} - x)_S}^2] > \frac{|S|\mu^2}{R} \right ] \\
\label{eq:prpsets2} &\le \frac{\Var(\norm{2}{(\wh{x}_{\text{tim}} - x)_S}^2)}{(|S| \mu^2/R)^2}.
\end{align}
Thus we need to bound
\begin{equation} \label{eq:prpsets3}
\Var(\norm{2}{(\wh{x}_{\text{tim}} - x)_S}^2) = \sum_{i \in S} \Var(E_{\text{tim},i}^2) + \sum_{i_1 \ne i_2 \in S} \Cov(E_{\text{tim},i_1}^2,E_{\text{tim},i_2}^2).
\end{equation}
For each $i$ we have $\Var(E_{\text{tim},i}^2) \le \E[E_{\text{tim},i}^4] \lesssim \mu^4/R^2$
by Corollary~\ref{cor:timfimvariance}.  The covariance term is the harder part to control.
Fix two coordinates $i_1 \ne i_2 \in S$ and consider FIM Count-Sketch with
respect to those two coordinates.  For shorthand write $E_{\text{tim},j}$ for $E_{\text{tim},i_j}$
($j=1,2$) and define $E_{\text{fim},j} =  (\wh{x}_{\text{fim}} - x)_{i_j}$.  Then
\[\begin{split} \Cov(E_{\text{tim},1}^2,E_{\text{tim},2}^2) = & \Cov(E_{\text{fim},1}^2,E_{\text{fim},2}^2) \\ & \quad + \Cov(E_{\text{fim},1}^2,E_{\text{tim},2}^2-E_{\text{fim},2}^2) + \Cov(E_{\text{fim},2}^2,E_{\text{tim},1}^2-E_{\text{fim},1}^2) \\ & \quad + \Cov(E_{\text{tim},1}^2-E_{\text{fim},1}^2,E_{\text{tim},2}^2-E_{\text{fim},2}^2). \end{split}\]
The first term vanishes because, by construction, $(\wh{x}_{\text{fim}})_{i_1}$ and $(\wh{x}_{\text{fim}})_{i_2}$
are independent.  We shall bound the remaining terms using the Cauchy-Schwarz inequality
$\Cov(X,Y) \le \sqrt{\Var(X)\Var(Y)}$.  

By Corollary~\ref{cor:timfimvariance} we have
$\Var(E_{\text{fim},j}^2) \le \E[E_{\text{fim},j}^4] \lesssim \mu^4/R^2$.
By Proposition~\ref{prp:fimtim}, $\Var(E_{\text{tim},j}^2-E_{\text{fim},j}^2) \lesssim \mu^4/(R^2 k^{1/7})$.
Thus
\[\Cov(E_{\text{tim},1}^2,E_{\text{tim},2}^2) \lesssim 2 \cdot \frac{\mu^4}{R^2 k^{1/14}} + \frac{\mu^4}{R^2 k^{1/7}} \lesssim \frac{\mu^4}{R^2 k^{1/14}}.\]
Substituting back into (\ref{eq:prpsets3}),
\[\Var(\norm{2}{(\wh{x}_{\text{tim}} - x)_S}^2) \lesssim |S| \cdot \frac{\mu^4}{R^2} + |S|(|S|-1) \cdot \frac{\mu^4}{R^2 k^{1/14}} \lesssim |S|^{2-\tfrac{1}{14}} \cdot \frac{\mu^4}{R^2}.\]
And finally, substituting this into (\ref{eq:prpsets2}), we get the desired bound.
\end{proof}

\end{document}

%% file: theorems.tex
\makeatletter
\newcommand{\define}[3]{%
  \@namedef{thm@#1}{#3}%
  \@namedef{thmtypen@#1}{lemma}%
  \newtheorem{thmtype@#1}[theorem]{#2}%
  \newtheorem*{thmtypealt@#1}{#2~\ref{#1}}%
}
\newcommand{\declare}[1]{%
  \begin{thmtype@#1} \label{#1}
    \@nameuse{thm@#1}
  \end{thmtype@#1}
}
\newcommand{\state}[1]{%
  \begin{thmtypealt@#1}
    \@nameuse{thm@#1}
  \end{thmtypealt@#1}
}
\makeatother